\documentclass{article}
\usepackage[utf8]{inputenc}
\usepackage{amsmath,amsthm}
\usepackage{amssymb}
\usepackage{xspace}
\usepackage{url}
\usepackage{graphicx}
\usepackage{complexity}
\usepackage{caption}
\usepackage{subcaption}
\usepackage{multirow} 
\usepackage{lscape} 
\usepackage[affil-it]{authblk} 
\usepackage{hyperref}
\usepackage{svg}
\usepackage{enumitem} 

\newtheorem{theorem}{Theorem}

\newtheorem{proposition}[theorem]{Proposition}
\newtheorem{claim}{Claim}[theorem]

\author{Allen Ibiapina \and Ana Silva}
\affil{ParGO Group - Parallelism, Graphs and Optimization\\ Departamento de Matemática - Centro de Ciências\\ Universidade Federal do Ceará, Brazil}

\date{}

\title{Snapshot disjointness in temporal graphs}

\usepackage{complexity}
\usepackage{xspace}

\newcommand{\snapcutproblem}{\textsc{\snapcut{s,z}}\xspace}
\newcommand{\snappathproblem}{\textsc{\snapdisjoint\xspace temporal $s,z$-paths}\xspace}

\newcommand{\N}{{\mathbb{N}\setminus\{0\}}}

\newcommand{\twalk}[1]{temporal $#1$-walk}
\newcommand{\tpath}[1]{temporal $#1$-path}

\newcommand{\snapdisjoint}{snapshot disjoint\xspace}
\newcommand{\snapcut}[1]{snapshot $#1$-cut\xspace}
\newcommand{\maxsnap}[2]{$sp_{#1}(#2)$\xspace}
\newcommand{\minsnap}[2]{$sc_{#1}(#2)$\xspace}

\begin{document}

\maketitle

\begin{abstract}
In the study of temporal graphs, only paths respecting the flow of time are relevant. In this context, many concepts of walks disjointness were proposed over the years, and the validity of Menger's Theorem, as well as the complexity of related problems, has been investigated. 
In this paper, we introduce and investigate a type of disjointness that is only time dependent. Two walks are said to be \emph{snapshot disjoint} if they are not active in a same snapshot (also called timestep). The related paths and cut problems are then defined and proved to be $\W[1]$-hard and $\XP$-time solvable when parameterized by the size of the solution. Additionally, in the light of the definition of Mengerian graphs given by Kempe, Kleinberg and Kumar in their seminal paper (STOC'2000), we define a \emph{Mengerian graph for time} as a graph $G$ that cannot form an example where Menger's Theorem does not hold in the context of snapshot disjointness. We then give a characterization in terms of forbidden structures and provide a polynomial-time recognition algorithm. 
Finally, we also prove that, given a temporal graph $(G,\lambda)$ and a pair of vertices $s,z\in V(G)$, deciding whether at most $h$ multiedges can separate $s$ from $z$ is $\NP$-complete.
\end{abstract}

\section{Introduction}
 A temporal graph can be described as a graph that varies in time. Such objects can be modeled in different ways, usually according to the application being considered, and have appeared in the literature under many names as for instance dynamic networks~\cite{XFJ.03}, temporal networks~\cite{Holme.15}, time-varying graphs~\cite{CFQS.12}, etc. For surveys we refer the reader to~\cite{Holme.15,LVM.18}. In this paper, we consider a \emph{temporal graph} to be a pair $(G,\lambda)$, where $G$ is a multigraph (hereon called just graph) and $\lambda$ is a function, called \emph{timefunction}, that relates each edge to a discrete label telling when such edge is going to be \emph{active}. The value $\max_{e\in E(G)}\lambda(e)$ is called \emph{lifetime} and is denoted by $\tau$. Also, graph $G$ is called the \emph{base graph}. See Figure~\ref{fig:exemplo_intro} for an example.

 \begin{figure}[h]
	\centering
	\begin{subfigure}[b]{0.4\textwidth}
         \centering
         \includegraphics[width = 5cm]{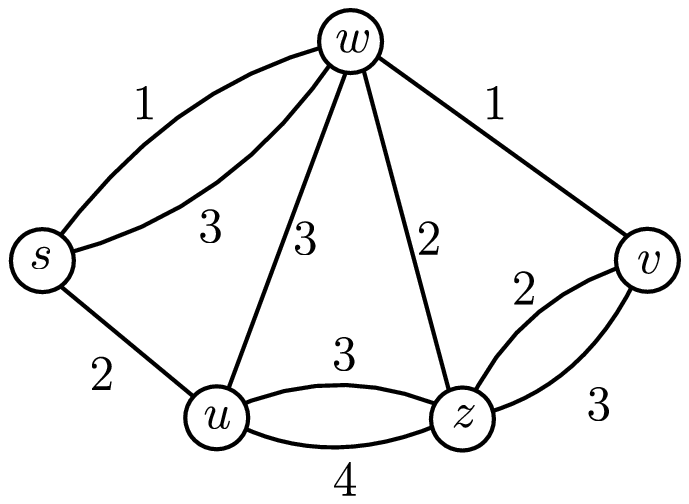}
         \caption{Temporal graph with lifetime~4.}
         \label{fig:exemplo_intro}
     	\end{subfigure}
     	\hfill
	\begin{subfigure}[b]{0.4\textwidth}
         \centering
         \includegraphics[width = 5cm]{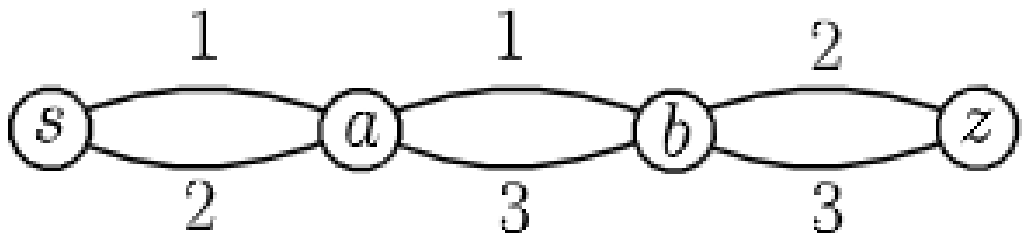}
         \caption{Pair $s,z$ whose cut must be bigger than the number of snapshot disjoint paths.}
         \label{fig:bigger_cut}
     	\end{subfigure}
     \caption{Examples.}\label{fig:intro}
\end{figure}

 Many practical problems are modeled in temporal graphs (see~\cite{Holme.15} for a nice collection of practical examples), and among the most common ones are those related to temporal walks and connectivity. A \emph{temporal walk} is a walk that respects the flow of time; for simplicity, we represent them as sequences of vertices and timesteps. For instance, in Figure~\ref{fig:exemplo_intro}, the sequence $(w,2,z,2,v,3,z,4,u)$ is a temporal walk between $w$ and $u$ (also called temporal $w,u$-walk). A \emph{temporal path} is then defined as a temporal walk whose internal vertices are all distinct. Hence, the previously mentioned temporal $w,u$-walk is not a path, while the walk $(w,2,z,4,u)$ is a temporal path. Additionally, some authors deal only with walks and paths whose edges are active in strictly increasing times; in such case, $(s,1,w,2,z,3,v)$ is a valid temporal $s,v$-path, while $(s,2,u,3,z,3,v)$ is not. To distinguish from these, we say that a walk/path is \emph{strict} if the edges are active in strictly increasing times, and that it is \emph{non-strict} if they are active in non-decreasing times.

 In contrast with classic graph theory, when dealing with walks and paths in temporal graphs, it is not always the case that the problems defined on walks are equivalent to those defined on paths. See for instance~\cite{AFG.17,CHMZ.19,FMNR.22,ILMS.arxiv}. This is not the case for the problems investigated here, and
 this is why we interchangeably use the terms walks and paths.

 Connectivity problems concern the \emph{robustness} of a network, which translates into knowing how many independent (or disjoint) ways there are  to go from one vertex to another, and how easy it is to break such connections. In this paper, we introduce a new robustness concept that relies on the time aspect of a network. To better understand these concepts, consider the following scenario.
 Suppose a temporal graph $(G,\lambda)$ models a communication network. Such network might be prone to interruptions of all communications at a given timestep due to attacks, blackouts, maintenance, etc. A good measure of robustness of such networks could then be the minimum number of timesteps in which the communications must get interrupted in order to break all possible connections between a pair of vertices. A network with higher measure means that it is less susceptible to failing under such interruptions and hence is considered more robust. In Figure~\ref{fig:exemplo_intro}, for instance, if there is an interruption on timesteps $2$ and~$3$, then vertex $s$ cannot relay a message to $z$ anymore, while it still can relay messages to $v$ through the path $(z,1,w,1,v)$. 

 To model such scenario, we say that two temporal  $s,z$-paths $P$ and $Q$ are \emph{\snapdisjoint} if, at any given timestep, at most one between $P$ and $Q$ is traversing any edges. For example, in Figure~\ref{fig:exemplo_intro}, paths $(s,1,w,1,v,2,z)$ and $(s,3,w,3,u,3,z)$ are two \snapdisjoint temporal $s,z$-paths. We also say that a set $S$ of timesteps is a \emph{\snapcut{s,z}} if every temporal $s,z$-path uses an edge active in timestep $i$ for some $i\in S$. For example, in Figure~\ref{fig:exemplo_intro}, $S=\{2,3\}$ is a \snapcut{s,z}. The following problems are then defined.

  \vspace{0.3cm}

        \noindent $\le h$-\snapcutproblem\\
        \textbf{Input.} A temporal graph $(G,\lambda)$, vertices $s,z\in V(G)$, and an integer $h$.\\
        \textbf{Question.} Is there a \snapcut{s,z} in $(G,\lambda)$ of size at most $h$?

        \vspace{0.3cm}
        
        \noindent $\ge k$-\snappathproblem\\
        \textbf{Input. }A temporal graph $(G,\lambda)$, a pair of vertices $s,z\in V(G)$, and a positive integer $k$.\\
        \textbf{Question. }Is there a set of \snapdisjoint  temporal $s,z$-paths in $(G,\lambda)$ of size at least $k$?
        \vspace{0.1cm}

We prove that, when parameterized by $h$ and $k$ respectively, both problems are $\W[1]$-hard, and that this is best possible, i.e., that they are also $\XP$. While the $\XP$ algorithm for \snapcutproblem follows easily from the definition and the fact that we can test all possible cuts in $\XP$ time (namely, $O(\tau^h)$ time), the algorithm for \snappathproblem is much more involved and uses a technique similar to the one applied to find $k$ vertex disjoint paths between $k$ given pairs of vertices (also known as the \textsc{$k$-linkage problem}) in a DAG~\cite{SP.78}. 
 As we will see in the related works, this is the first result of such kind, with all previously defined disjointness either having the related paths problem polynomial-time solvable or para-$\NP$-complete (i.e., $\NP$-complete for fixed values of $k$).

 A celebrated result in classic graph theory tells us that, in a graph $G$ and for every pair $s,z\in V(G)$, the maximum number of internally vertex disjoint $s,z$-paths is equal to the minimum size of an $s,z$-cut (vertices whose removal breaks all $s,z$-paths). This is the well known Menger's Theorem, and it holds on both undirected and directed graphs, as well as for edge-disjoint paths and edge cuts. 
 When translating these concepts to temporal graphs, it is natural to ask whether a version of Menger's Theorem holds. The answer in our context is no, as can be witnessed by the example in Figure~\ref{fig:bigger_cut}. Note that any two temporal $s,z$-paths intersect in some timestep, while there is no \snapcut{s,z} of size~1. Indeed, $(s,1,a,1,b,2,z)$ does not use edges active in timestep $3$, $(s,2,a,3,b,3,z)$ does not use~1, and $(s,1,a,1,b,3,z)$ does not use~2.

 In their seminal paper, Kempe, Kleinberg and Kumar~\cite{KKK.00}, in the context of vertex disjoint temporal paths, defined a Mengerian graph as being a graph where Menger's Theorem would hold for whatever choice of timefunction. They then characterize these graphs when constrained to simple graphs (every multiedge has multiplicity~1), and more recently their result was generalized to allow for multigraphs~\cite{IS.arxiv}. Here, we say that $G$ is \emph{Mengerian for time} if, for every timefunction $\lambda$ and every $s,z\in V(G)$, the maximum number of \snapdisjoint temporal $s,z$-paths in $(G,\lambda)$ is equal to the minimum size of a \snapcut{s,z}. In other words, the \snapdisjoint version of Menger's Theorem always holds on temporal graphs whose base graph is $G$.  We then give the following characterization. The formal definition of an m-topological minor is given in Section~\ref{sec:defs}, but for now it suffices to say that, when subdividing an edge $e$, the multiplicity of the obtained edges is the same as $e$.

	\begin{figure}[h]
		\centering
			\includegraphics[scale=0.65]{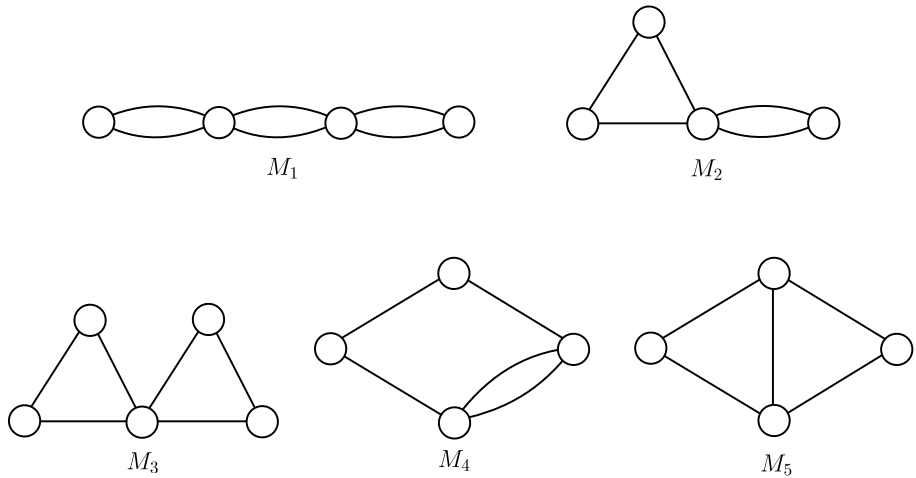}
		\caption{Graphs in the set $\mathcal{M}$.}
		\label{fig:todos}
	\end{figure}

 \begin{theorem}\label{thm:Menger}
 Let $G$ be a graph. Then $G$ is Mengerian for time if and only if $G$ does not have any of the graphs in Figure~\ref{fig:todos} as m-topological minor. Moreover, we can recognize whether $G$ is Mengerian for time in polynomial time.
 \end{theorem}

 Finally, in order to fill an open entry related to multiedge disjoint temporal paths (the definition is presented shortly), we prove in Section~\ref{sec:multiedge_cut} that the related cut problem is $\NP$-complete even if the temporal graph has lifetime equal to~2.

\paragraph*{Related problems. }
 As snapshot disjointness is a newly introduced concept, no previous results exist. We then refer the reader to the many results about Menger's related concepts in temporal graphs. In this context, the vertex disjoint version of Menger's Theorem was proved not to hold by Berman~\cite{B.96}. Since then a number of papers have investigated the complexity of related problems~\cite{KKK.00,ZFMN.20}, as well as new structural concepts like the definition of Mengerian graphs~\cite{KKK.00,IS.arxiv}, and adaptations to temporal vertex disjoint versions~\cite{MMS.19,ILMS.arxiv}. Because our problem is more closely related to edge connectivity, we refrain from commenting in details the results on vertex connectivity, but refer the reader to~\cite{ILMS.arxiv} for an overview of such results. In what follows, we present the edge-related concepts and existing results. These are summarized in Table~\ref{table:results}.

\begin{table}[hbt]
\footnotesize
\centering
\begin{tabular}{c|c|c|c|c|}
			
\multirow{2}{*}{ } &  \multicolumn{2}{|c|}{{\tt Non-strict}} & \multicolumn{2}{|c|}{{\tt Strict}} \\\cline{2-5}
			
		& {\tt $\ge k$-Walks}                   & {\tt $\le h$-Cut}                        & {\tt $\ge k$-Walks}                   & {\tt $\le h$-Cut} \\
\hline
{\tt Multiedge} & $\NP$c \cite{B.96}, if $G$ dir.,& $\NP$c for $\tau=2$           & $\NP$c for $\tau=5$ \cite{IPS.82}  & $\NP$c for $\tau=4$~\cite{B.etal.10} \\
		& even for $k=\tau=2$             & (Theorem~\ref{thm:multiedgecut})  &   and $k=2$~\cite{LMS.90}       & and $\W[1]$ for $h$~\cite{GT.11}\\
			\hline
{\tt Edge}     & \multicolumn{2}{|c|}{Polynomial~\cite{B.96}} & \multicolumn{2}{|c|}{Polynomial~\cite{MMS.19}} \\
\hline
{\tt Snapshot} &  $\W[1]$ for $k$ (Th.~\ref{theo:snappaths})    & $\W[1]$ for $h$ (Th.~\ref{thm:snapcut_W1hard})        &  \multicolumn{2}{|c|}{Open}\\
               &  $\XP$ for $k$ (Th.~\ref{thm:positive_paths}) &  $\XP$ for $h$ (Th.~\ref{theorem:cuts}) & \multicolumn{2}{|c|}{}\\
\hline
{\tt Node dep.} & \multicolumn{2}{|c|}{Open} & \multicolumn{2}{|c|}{Polynomial~\cite{MMS.19}}\\

			\hline
		\end{tabular}
		\caption{On the leftmost column, we specify the type of disjointness. Above, $\tau$ denotes the lifetime of the temporal graph, $k$ denotes the number of paths, $h$ denotes the size of the cut, $\NP$c stands for $\NP$-completness, and $\W[1]$ or $\XP$ stands for $\W[1]$-hardness or $\XP$ results when parameterizing by the size of the solution. Gray cells are proved in this paper.}\label{table:results}
	\end{table}

A set of temporal $s,z$-walks are \emph{edge disjoint} if they share no edges, and are \emph{multiedge disjoint} if they share no multiedges. For example, in Figure~\ref{fig:exemplo_intro}, the paths $(s,1,w,2,z)$ and $(s,3,w,3,u,3,z)$ are edge disjoint, but are not multiedge disjoint, since they share the multiedge with endpoints $sw$. A set of (multi)edges is a \emph{temporal (multi)edge $s,z$-cut} if they intersect every temporal $s,z$-walk. For example, in Figure~\ref{fig:exemplo_intro}, $sw$ and $su$ form a multiedge $s,z$-cut, but if we want an edge $s,z$-cut, we have to pick both edges whose endpoints are $sw$. In some works, instead of using a multigraph, it is used a simple graph together with a timefunction that assigns to each edge a finite set of positive integers. In such cases, an edge in our context would be called a temporal edge, while a multiedge in our context would be a simple edge of the base graph. 

In~\cite{B.96}, Berman showed that the edge problems for non-strict temporal paths are polynomial-time solvable, and that deciding the existence of at least $k$ multiedge disjoint temporal paths is $\NP$-complete, $G$ directed or undirected, and if $G$ is directed, then the same holds even if $k=\tau=2$. Up to our knowledge, no result concerning the cut problem related to multiedges was presented so far. By a simple modification of a proof in~\cite{ZFMN.20}, we present in Section~\ref{sec:multiedge_cut} a proof of $\NP$-completeness of the multiedge cut problem. Our proof works also for the case where $G$ is a directed multigraph. Concerning strict paths, the complexities of these problems follow directly from results about problems on bounded length paths~\cite{B.etal.10,GT.11}. Additionally, the strict problems related to edge disjoint paths was shown to be polynomial-time solvable in~\cite{MMS.19}. We mention that, among all the problems appearing in Table~\ref{table:results}, Menger's Theorem holds only for edge disjoint paths in both the strict and non-strict contexts~\cite{B.96,MMS.19}, and node departure disjoint strict paths, defined below.

 Another related concept is that of \emph{node departure disjoint}, introduced in~\cite{MMS.19}. Given a temporal graph $(G,\lambda)$ with lifetime $\tau$, a set of strict temporal $s,z$-walks is \emph{node departure disjoint} if no two of these paths leave a vertex $u$ in the same timestep. For example, in Figure~\ref{fig:exemplo_intro}, $(s,1,w,2,z)$ and $(s,3,w,3,u,3,z)$ are node departure disjoint. Additionally, a set $S\subseteq V(G)\times [\tau]$ is a \emph{node departure $s,z$-cut} if all strict temporal $s,z$-walks contains an edge departing from $u$ in time $t$, for some $(u,t)\in S$. For example, in Figure~\ref{fig:exemplo_intro}, the set $S = \{(s,1),(s,2),(s,3)\}$ is a node departure $s,z$-cut. 
 In~\cite{MMS.19}, the authors prove that the maximum number of node departure disjoint $s,z$-walks is equal to the minimum size of a node departure $s,z$-cut. Even though the authors do not comment on the complexity of the related problems, their proof leads to a polynomial time algorithm as it consists of building a flow network and proving that the searched values are equivalent to applying the famous Maxflow-Mincut Theorem. Up to our knowledge, their results have not been investigated for the non-strict context.

The text is organized as follows. In Section~\ref{sec:defs}, we present definitions, terminology and some basic results. In Section~\ref{sec:positive}, we present our \XP\space algorithms. In Section~\ref{sec:negative}, we prove that $\le h$-\snapcutproblem and $\ge k$-\snappathproblem are \W[1]-hard when parameterized by $h$ and $k$, respectively.  In Section~\ref{sec:menger}, we characterize Mengerian graphs. Finally, in Section~\ref{sec:multiedge_cut}, we prove that $\le h$-Multiedge cut is \NP-complete, and in Section~\ref{sec:conclusion} we present our concluding remarks.


\section{Definitions and Terminology}\label{sec:defs}

Given positive integers $i,j \in \mathbb{N}$ such that $j\geq i$, we denote by $[i,j]$ the set $\{i,i+1,\dots,j\}$ and by $[j]$ the set $\{1,\dots,j\}$.

A \emph{graph} is a triple $(V,E,f)$ where $V$ and $E$ are finite sets that we call \emph{vertex set} and \emph{edge set} respectively, and $f$ is a function that, for each $e \in E$ associates a pair $xy$ of elements in $V$, where $x\neq y$. We say that edge $e$ is \emph{incident} to $x$ and $y$, that $x,y$ are the \emph{endpoints} of $e$, and that $e$ \emph{connects} $x$ and $y$. We omit $f$ in the rest of the paper and refer simply to the endpoints of $e$ instead. We also call the pair $xy$ a \emph{multiedge}, and the number of edges with endpoints $xy$ is the \emph{multiplicity} of the multiedge $xy$. 
If the multiplicity of each edge is~1, we say that $G$ is a \emph{simple graph}. We denote by $U(G)$ the simple graph obtained from $G$ by decreasing the multiplicity of all multiedges to~1.
See~\cite{west.book} for further basic definitions of graph theory.

Given a  graph $G$ and a set of vertices $Z\subseteq  V(G)$, the \emph{identification of $Z$} is the graph obtained from $G - Z$ by adding a new vertex $z$ and, for every edge $e$ with endpoints $z'u$ where $z' \in Z$ and $u\notin Z$, add an edge $e'$ with endpoints $zu$. The graph $G'$ obtained from $G$ by a \emph{subdivision} of an edge $e$ with endpoints $uv$ is the graph having $V(G)\cup \{z_e\}$ as vertex set, and $E(G-e)\cup\{e', e''\}$ as edge set, where $e'$ has endpoints $uz_e$ and $e''$ has endpoints $z_ev$. Finally, the graph obtained from $G$ by an \emph{m-subidivision} of a multiedge $xy$ is the graph obtained by subdividing all the edges with endpoints $xy$ and then identifying the new vertices. Observe Figure~\ref{fig:defmmenor} for an illustration of these definitions. The definition of m-subdivision has been introduced in~\cite{IS.arxiv}. Given a graph $H$, if $G$ has a subgraph that can be obtained from m-subdivisions of $H$, then we say that $H$ is an \emph{m-topological minor} of $G$.

	\begin{figure}[h]
		\centering
		\includegraphics[width = 8cm]{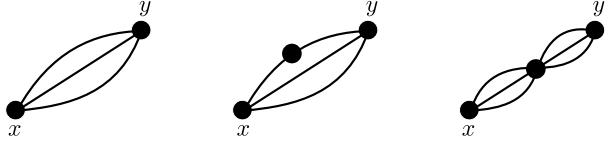}
		\caption{From left to right: the multiedge $xy$, the subdivision of an edge with endpoints $xy$, and the m-subdivision of $xy$.}
		\label{fig:defmmenor}
	\end{figure}

A \emph{temporal graph} is a pair $(G,\lambda)$ where $G$ is a graph and $\lambda \colon E(G) \to \N$. Depending on the context we refer to the elements of $\N$ as \emph{timesteps}. If an edge $e$ is such that $\lambda(e)= \alpha$ we say that $e$ is \emph{active} or \emph{appears} at timestep $\alpha$. A \emph{\twalk{x_1,x_q}} in $(G,\lambda)$ is a sequence $P$ that alternates vertices and edges $(x_0,e_1,x_1,\dots, e_{q},x_q)$ such that for every $i \in [q]$, $e_i$ is an edge between $x_i$ and $x_{i-1}$ and $\lambda(e_1)\leq \dots \leq \lambda(e_q)$. If $x_i\neq x_j$ for every $i,j \in [q]$ with $i\neq j$, we say that such temporal walk is a \emph{temporal path}. Moreover, we define $V(P)=\{x_1,\dots,x_q\}$ and $E(P)=\{e_1,\dots,e_q\}$. For our purposes, we can assume that the subgraph active at a given timestep is simple, i.e., that if $e$ and $e'$ have both endpoints $xy$, then $\lambda(e)\neq \lambda(e')$. Such assumption allows us to define a path as a sequence of vertices and timesteps $(x_0,t_1,x_1,\dots,t_{q-1},x_q)$ such that, for each $i \in [q]$, there is an edge connecting $x_{i} x_{i-1}$ active at timestep $t_i$. The \emph{lifetime} of $(G,\lambda)$ is denoted by $\tau(\lambda)$ and is the maximum integer such that there is an edge of $G$ active at such timestep. For each timestep $i \in \N$, the \emph{$i$-th snapshot} of $(G,\lambda)$ is the subgraph of $G$ defined as $H=(V(G),\lambda^{-1}(i))$. 


Let $(G,\lambda)$ be a temporal graph with lifetime $\tau$. Also, let $s,z \in V(G)$ be vertices in $G$ and $Q,J$ \tpath{s,z}s. We say that $Q$ and $J$ are \emph{\snapdisjoint}\space if $\lambda(E(Q)) \cap \lambda(E(J)) = \emptyset$. A subset $S \subseteq [\tau]$ is a \emph{\snapcut{s,z}} if every \tpath{s,z} uses an edge active at some timestep in $S$.
We denote by \maxsnap{G,\lambda}{s,z} the maximum number of \snapdisjoint\space\tpath{s,z}s and by \minsnap{G,\lambda}{s,z} the minimum size of a \snapcut{s,z}. Observe that if the above definitions are made in terms of temporal paths, then these parameters would not change. 
A graph $G$ is \emph{Mengerian (for time)} if, for every timefunction $\lambda$ on $E(G)$, and every $s,z \in V(G)$, $s\neq z$, we have that \maxsnap{G,\lambda}{s,z}$=$\minsnap{G,\lambda}{s,z}. 
The following will be useful later, and the proofs can be found in Appendix~\ref{app:propositions}.

\begin{proposition}\label{prop:nonmeng}
    If $G$ is non-Mengerian, then an m-subdivision of $G$ is also non-Mengerian. 
\end{proposition}

\begin{proposition}
    \label{prop:subgraphClosed}
		$G$ is Mengerian if and only if $H$ is Mengerian, for every $H\subseteq G$. 
\end{proposition}


 \section{Positive Results}\label{sec:positive}

        In this section, we give $\XP$ algorithms for both \snapcutproblem and \snappathproblem. Given the results of Section~\ref{sec:negative}, unless $\FPT=\W[1]$-hard, $\XP$ algorithms are best possible from the point of view of parameterized complexity. The first algorithm is quite simple and consists of the usual approach of testing all possible cuts.

        \begin{theorem}\label{theorem:cuts}
            Given a temporal graph $(G,\lambda)$ of lifetime $\tau$,  a positive integer $h$ and $s,z \in V(G)$, we can solve $\le h$-\snapcutproblem in  $O(\tau^{h}\cdot (\lvert V(G)\rvert +\lvert E(G)\rvert))$. 
        \end{theorem}

        The next algorithm is much more involved, and uses a technique similar to the one used to find disjoint paths between given pairs of vertices in a DAG~\cite{SP.78}.

        \begin{theorem}\label{thm:positive_paths}
            Given a temporal graph $(G,\lambda)$, vertices $s,z \in V(G)$ and a positive integer $k$, we can solve \snappathproblem in time $O(m^k)$, where $m = \lvert E(G)\rvert$.
        \end{theorem}
        \begin{proof}
          We construct a digraph $D$ with vertices $s^*$ and $z^*$ such that $|V(D)|=O(m^k)$ and there is an $s^*,z^*$-path in $D$ if and only if there are $k$ \snapdisjoint\space \tpath{s,z}s in $(G,\lambda)$.

          The vertex set of digraph $D$ is equal to the $k$-tuples formed by edges of $G$, together with vertices $s$ and $z$; formally $V(D)\subseteq F^k$, where $F = E(G)\cup \{s,z\}$. Vertex $s^*$ is set to be equal to $(s,\ldots,s)$, while vertex $z^*$ is set to be equal to $(z,\ldots,z)$. Each dimension of $V(D)$ represents one of the desired $k$ disjoint paths, and a set of snapshot disjoint temporal $s,z$-paths $P_1,\ldots,P_k$ will be represented by an $s^*,z^*$-path $P$ in $G$, as previously said. So $s^*$ represents the starting point, and $z^*$ represents the finish point of every temporal $s,z$-path. Then, when an edge of $D$ is traversed by $P$, we want that one of the $k$ paths also traverses an edge. Because we want to allow that only one of the paths gets closer to $z$ with each step of $P$, there will be an edge from $\alpha \in V(D)$ to $\beta\in V(D)$ only if exactly one position of $\alpha$ and $\beta$ differ. Not only this, but we want that, at each step of $P$, the path $P_i$ that gets closer to $z$ is the one whose last traversed edge occurs the earliest among all the $P_i$'s. In the next paragraph, we formally construct digraph $D$.

          As previously said, let $F = E(G)\cup \{s,z\}$. Along the construction, we will be referring to Figures~\ref{fig:positive_graph_relationM} and~\ref{fig:positive_digraph}. Because we want to avoid simultaneous traversal of paths that intersect in a snapshot, we only consider elements of $F^k$ whose edges of $G$ all appear in distinct snapshots. Indeed, if $k=2$ and we allow for instance the existence of vertex $(e,e')$ such that $t=\lambda(e) = \lambda(e')$, then this would mean that the constructed paths $P_1$ and $P_2$ intersect in timestep $t$. Therefore, we define $V$ as formalized below. Observe that this implies, in Figure~\ref{fig:positive_digraph}, that vertices $\{(e,e)\mid e\in E(G)\}\cup \{(az_2,sb),(sb,az_2)\}$ do not exist in $V(D)$, where $az_2$ denotes the edge with endpoints $az$ active in timestep~2.

          \[V =  \{(u_1,\dots,u_k) \in F^k \mid \forall i,j \in [k]\mbox{ with $i\neq j$, we have }\lambda(u_i)\neq \lambda(u_j)\mbox{ or } u_i=u_j \in \{s,z\} \} \]

          \begin{figure}[h]
  \centering
  \begin{subfigure}[b]{0.4\textwidth}
         \centering
         \includegraphics[width = 4cm]{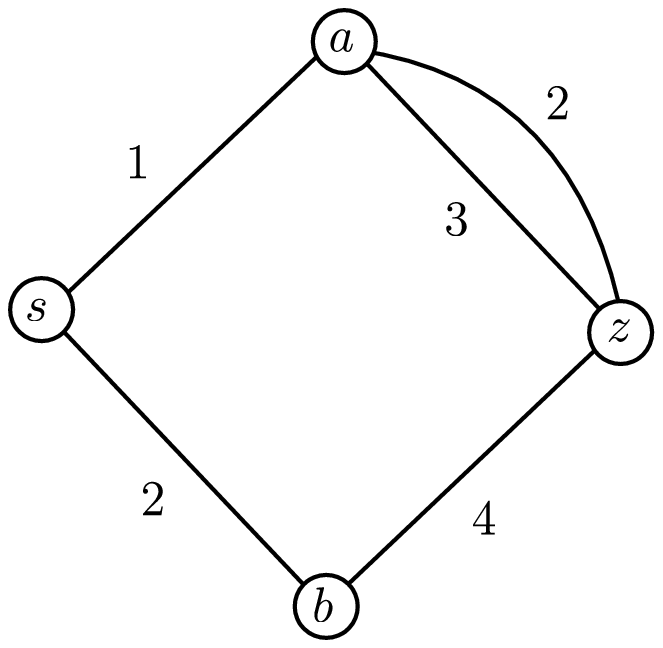}
         \caption{Example of temporal graph $(G,\lambda)$.}
         \label{fig:positive_graph}
      \end{subfigure}
      \hfill
  \begin{subfigure}[b]{0.4\textwidth}
         \centering
         \includegraphics[width = 4cm]{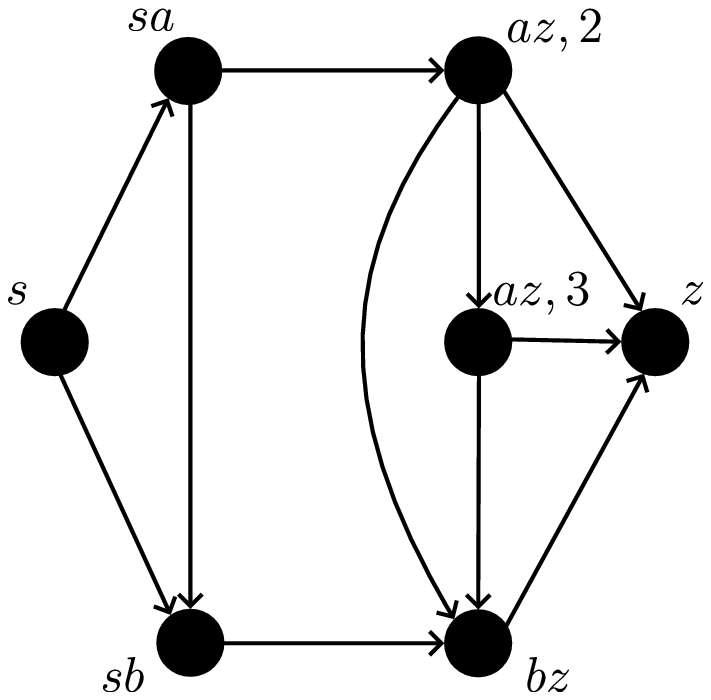}
         \caption{Auxiliary graph $M$ on the set $F = E(G)\cup \{s,z\}$.}
         \label{fig:positive_relationM}
      \end{subfigure}
     \caption{Example of construction in Theorem~\ref{thm:positive_paths}.}\label{fig:positive_graph_relationM}
\end{figure}

Now, we define the edge set of $D$. For this, we first construct an auxiliary graph $M$ whose vertex set is equal to $F$; observe Figure~\ref{fig:positive_relationM} to follow the construction. First of all, we want that a traversal of an edge in $D$ translates into a valid traversal in $(G,\lambda)$. Therefore, for every pair $e,f\in F$, add to $M$ an edge from $e$ to $f$ only if $e$ can be followed by $f$ in a \tpath{s,z} in $(G,\lambda)$. Formally, add $ef$ in the following cases: for every $e\in E(G)$, and every $f\in E(G)$ adjacent to $e$ and such that $\lambda(e)\le \lambda(f)$; for $e=s$ and every $f\in E(G)$ incident to $s$; and for every $e\in E(G)$ incident to $z$ and $f=z$.  

Finally, as previously said, we want that at each edge traversal of an $s^*,z^*$-path in $D$, the path in $(G,\lambda)$ that is getting closer to $z$ is that one whose last used edge is the earliest among all the other paths. To help with this, we also define $\lambda(s)$ to be equal to $0$, and $\lambda(z)$ to be equal to $\tau+1$, where $\tau$ is equal to the lifetime of $(G,\lambda)$. This means intuitively that we give always priority to leave $s$, and that, once we reach $z$ in any dimension, then we cannot depart from $z$ anymore. So, given a vertex $\alpha=(u_1,\ldots,u_k)\in V(D)$, we add an edge from $\alpha$ to $\beta\in V(D)$ if and only if: 
\begin{itemize}
 \item $\beta$ differ from $\alpha$ in exactly one position, $i$; 
 \item $i$ is such that $\lambda(u_i)\le \min_{j\in [k]}\lambda(u_j)$; and
 \item By letting $u'_i$ be the value in the $i$-th position of $\beta$, we have that $u_iu'_i$ is a valid move, i.e., that $u_iu'_i\in E(M)$. 
\end{itemize}

Observe Figure~\ref{fig:positive_digraph}. Another way of seeing this construction is by starting with copies of $M$ on each row and column of $D$, then removing the vertices that do not belong to $D$, and finally removing from row/column $e$ any edge leaving $f$ with $\lambda(f)>\lambda(e)$.

\begin{figure}[h]
    \centering
    \includegraphics[width = 10cm]{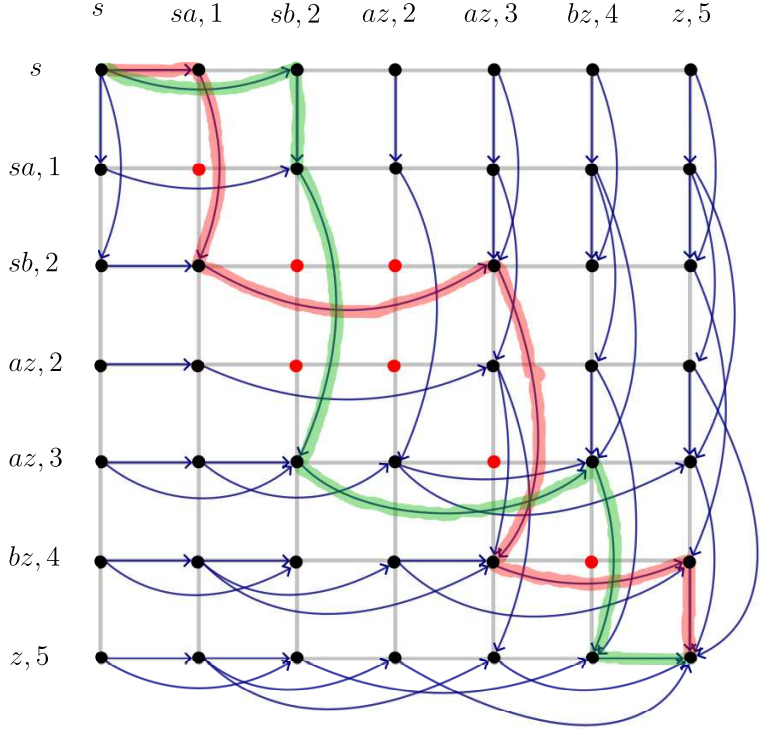}
    \caption{Digraph $D$ related to the temporal graph in Figure~\ref{fig:positive_graph}; value $k=2$ is being used, which means that $V(D)\subseteq F^2$. Each row and column is labeled with an element $e$ of $F$, together with the value $\lambda(e)$; this will help in the construction. A vertex $(e,f)$ of $D$ is represented in the intersection of row $e$ and column $f$. Red dots in the figure represent the fact that the related pair $(row,column)$ is not  a vertex in $D$. The snapshot  disjoint \tpath{s,z}s $P_1 = (s, 1, a, 3,z)$ and $P_2 = (s,2,b,4,z)$ can be obtained either through the red or the green $s^*,z^*$-path.}
    \label{fig:positive_digraph}
  \end{figure}

  Now, we prove that there are $k$ \snapdisjoint\space \tpath{s,z}s in $(G,\lambda)$ if and only if there is an $s^*,z^*$-path in $D$. 
  In what follows, given a vertex $\alpha\in V(D)$, we denote by $(a_1,\ldots,a_k)$ the tuple related to $\alpha$. 
  Recall that $s^*=(s,\dots,s)$ and $z^*=(z,\dots,z)$. Suppose $P_1,\dots, P_k$ is a set of \snapdisjoint\space \tpath{s,z}s in $(G,\lambda)$. 
  For each $i \in [k]$, let $e^1_i,\dots,e^{p_i}_i$ be the sequence of edges used in $P_i$ in order of traversal and define $e^0_i = s$ and $e^{p_i + 1}_i=z$. By induction, we define a sequence of vertices of $D$, $(s^*=\alpha^1,\dots,\alpha^q=\alpha)$, that forms an $s^*,\alpha$-path for some $\alpha$ with the following property:

  \begin{itemize}
    \item[(P)] For each dimension $i \in [k]$ and each $\ell \in [q]$, the sequence of edges traversed in dimension $i$ is a subpath of $P_i$. Formally, by removing $s$ and repetitions of edges from the sequence $(a_i^1,\ldots,a_i^q)$, we obtain a subsequence of $e^1_i,\dots,e^{p_i}_i$.
  \end{itemize}

  First, we define $\alpha_1 = s^*$; clearly property (P) holds as all paths start in $s$. Now suppose that sequence $\alpha_1,\ldots,\alpha_q$ satisfying Property (P) is obtained, $q\ge 1$. Let $i \in [k]$ be such that $\lambda(a^q_i) = \min_{j \in [k]} \lambda(a^q_j)$. 
  By Property (P), observe that either $a^q_i = s$, or $a^q_i$ is an edge of $P$, or $a^q_i=z$. If the latter occurs, then we get that $P$ is an $s^*,z^*$-path in $D$, since $\lambda(z) > \lambda(e)$ for every $e\in F\setminus \{z\}$, i.e., the only way $\lambda(z)$ is minimum is if all other positions are also equal to $z$. 
  So suppose one of the other cases occurs. Note that it means that there is some edge following $a^q_i$ in $P$, say $e^\ell_i$. By definition of temporal path, $\lambda(e^\ell_i)\ge \lambda(a^q_i)$; hence $a^q_ie^\ell_i\in E(M)$. 
  Define $\alpha^{q+1}$ to be equal to $\alpha^q$ except that in position $i$ we have $e^\ell_i$ instead of $a^q_i$, and note that $(\alpha^1,\ldots, \alpha^{q+1})$ is a path in $D$ that satisfies Property (P).

  Now suppose the existence of an $s^*,z^*$-path in $D$, $(\alpha^1,\dots,\alpha^q)$. We construct a set of $k$ \snapdisjoint\space \tpath{s,z}s in $(G,\lambda)$. 
  For this, for each $i\in [k]$, let $P_i$ be a sequence of edges obtained from dimension $i$, i.e., from $(a_i^1,\ldots,a_i^q)$ by removing occurrences of $s$ and $z$, and repetitions of edges. Because each transition respects $M$, we trivially get that $P_i$ defines a \tpath{s,z} in $(G,\lambda)$. It remains to show that such paths are snapshot disjoint. 
  Suppose otherwise, and let $i,j$ be such that there are edges $e_i$ in $P_i$ and $e_j$ in $P_j$ such that $\lambda(e_i)=\lambda(e_j)=\ell$. 
  Let $\ell_i$ be the smallest index such that $a_i^{\ell_i} = e_i$, and $\ell_j$ be the smallest index such that $a_j^{\ell_j}=e_j$. 
  By the definition of $V(D)$, we have that $\ell_i\neq \ell_j$. Indeed no vertex of $D$ can contain two elements of $F$ with same value of $\lambda$, and recall that $i\neq j$ as $P_i,P_j$ are distinct paths. So, we can suppose, without loss of generality, that $\ell_i < \ell_j$. 
  Observe that this means that $\alpha^{\ell_i-1}$ differ from $\alpha^{\ell_i}$ in exactly position $i$; additionally, it means that $\lambda(e_i)\le \min_{h\in [k]}\lambda(a^{\ell_i-1}_h)$. In particular, we have that $\ell=\lambda(e_i)\le \lambda(a^{\ell_i-1}_j)$. But observe that, in a fixed dimension, the values of $\lambda$ can only increase, i.e., since $\ell_i < \ell_j$, we get $\lambda(a^{\ell_i-1}_j)\le \lambda(a^{\ell_j}_j) = \lambda(e_j) = \ell$. We get a contradiction as in this case vertex $\alpha^{\ell_i-1}$ should not be defined as it contains two elements with the same value of $\lambda$, namely $a^{\ell_i-1}_j$ and $e_i$.

  To finish the proof just recall that $|V(D)|\leq (m+2)^k$, where $m = \lvert E(G)\rvert$, and that deciding if there is a path between two vertices in $D$ can be made in time $O(|V(D)|^2)$. So, deciding if there are $k$ \snapdisjoint\space \tpath{s,z}s in $(G,\lambda)$ can be done in time $O(m^k)$.
\end{proof}


\section{Negative Results}\label{sec:negative}

In this section, we prove that the algorithms presented in Section~\ref{sec:positive} are best possible, i.e., that $\ge k$-\snappathproblem\space and $\le h$-\snapcutproblem\space are \W[1]-hard when parameterized by $k$ and $h$, respectively. 

        \begin{theorem}\label{theo:snappaths}
            $\ge k$-\snappathproblem is $\W[1]$-hard when parameterized by $k$.
        \end{theorem}

        \begin{proof}
            We make a parameterized reduction from  \textsc{$\geq k$-Independent set} when parameterized by $k$. The input of such problem is a simple graph $G$ and an integer $k$, and the question is whether $G$ has an independent set of size at least $k$. This is known to be $\W[1]$-hard (see e.g.~\cite{param.book}).

            So consider an instance $G,k$  of \textsc{$\geq k$-Independent Set} and let $|V(G)|=n$. Observe Figure~\ref{fig:reducao_SnapPath} to follow the construction. First, add to $G'$ vertices $s$ and $z$. Then, for each $u\in V(G)$, add to $G'$ an $s,z$-path on $d(u)$ edges; denote such path by $Q_u$. Now, consider any ordering $e_1,\ldots, e_m$ of $E(G)$, and denote the edges incident to a vertex $u \in V(G)$ by $\delta(u)$. We can define $\lambda \colon E(G') \to \N$ in a way that each $Q_u$ is a temporal $s,z$-path using the orders of the edges in $\delta(u)$. Formally, for each $u\in V(G)$, let $\delta(u) = \{e_{i_1},\ldots,e_{i_q}\}$ with $i_1 < \ldots < i_q$, and define $\lambda(E(Q_u))$ to be equal to $\{i_1,\ldots,i_q\}$ in a way that $Q_u$ is a temporal path.

\begin{figure}[h]
    \centering
    \includegraphics[width = 12cm]{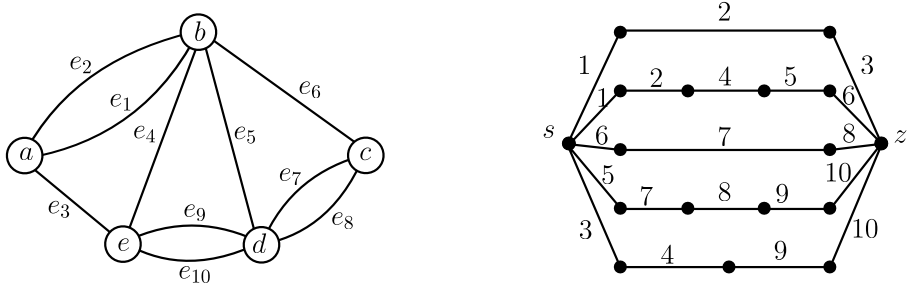}
    \caption{To the left, graph $G$, and to the right, the constructed temporal graph $(G', \lambda)$. In $G'$, paths $P_a,P_b,P_c,P_d,P_e$ are depicted from top to bottom, in this order.}
    \label{fig:reducao_SnapPath}
  \end{figure}

  Because of space constraints, the proof of correctedness is presented in Appendix~\ref{app:snappaths}.
  \end{proof}


        Now, we prove the analogous result for the cut problem.

        \begin{theorem}\label{thm:snapcut_W1hard}
            $\le h$-\snapcutproblem\space is $\W[1]$-hard when parameterized by $h$.
        \end{theorem}

        \begin{proof}
            We make a reduction from \textsc{Multicolored $k$-Clique}, when parameterized by $k$, known to be $\W[1]$-hard~\cite{param.book}. Such problem has as input a simple graph $G$, an integer $k$, and a partition of $V(G)$ into $k$ independent sets (alternatively, a  proper $k$-coloring), and the question is whether $G$ has a (multicolored) clique of size $k$. So let $G$ be a graph and $\{X_1,\dots,X_k\}$ be a proper $k$-coloring of $G$. By adding artificial vertices and edges if necessary, we can suppose that the number of edges between $X_i$ and $X_j$ is equal to a value $m$, for every pair $i,j \in [k]$. So, for $i,j\in [k]$, $i\neq j$, denote the set of such edges by $E_{i,j} = \{e^{i,j}_1,\dots,e^{i,j}_m\}$. We make this assumption in order to make presentation simpler. 
            
            Now, for each $i,j \in [k]$, $i\neq j$, we construct a gadget denoted by $F_{i,j}$. Observe Figure~\ref{fig:Fij} to follow the construction. 
            First add to $F_{i,j}$ the set of vertices $V_{i,j} = \{v^{i,j}_0,\dots,v^{i,j}_{2m}\}$, making the first $m+1$ of them form a path of multiplicity $m$, and the latter $m+1$ form a path of multiplicity~1. 
            Formally, for each $\ell\in \{0,\ldots,m-1\}$, add $m$ edges with endpoints $v^{i,j}_\ell v^{i,j}_{\ell+1}$. Also, for each $\ell\in \{m,\ldots,2m-1\}$, add 1 edge with endpoints $v^{i,j}_{\ell}v^{i,j}_{\ell+1}$. 
            Now, for each $\ell\in [m]$, add vertex $w^{i,j}_\ell$ and join such vertex with $v^{i,j}_\ell$ by a path with $m-1$ edges and denote such path by $P^{i,j}_\ell$. We say that vertex $w^{i,j}_\ell$ of our gadget is \emph{associated with edge $e^{i,j}_\ell$} of $E_{i,j}$. The timefunction is defined only later.

            \begin{figure}
                \centering
                \includegraphics[scale=0.7]{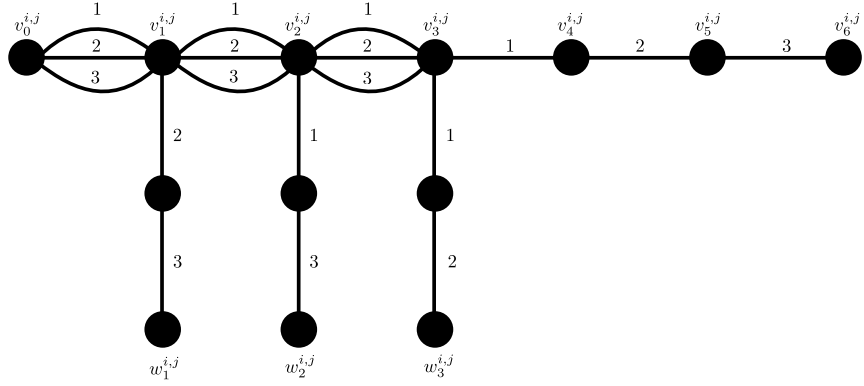}
                \caption{A representation of $F_{i,j}$ with labels of $\lambda$ where $m=3$ and $\Delta_{i,j}=\{1,2,3\}$.}
                \label{fig:Fij}
            \end{figure}

            Now, we finish the construction of our temporal graph. For this, take the union of all graphs $F_{i,j}$ and identify all vertices $v^{i,j}_{0}$, calling the obtained vertex $s$, and identify all vertices $v^{i,j}_{2m}$, calling the obtained vertex $z$. 
            Also, for each $i,j \in [k]$, $i\neq j$, and $\ell \in [m]$, we add two edges between $w^{i,j}_\ell$ and $z$. Denote by $G'$ the obtained graph, and by $W$ the set $\{w^{i,j}_\ell\mid i,j\in[k], i\neq j, \ell\in [m]\}$. Observe that $G'$ contains $O(k^2\cdot m)$ vertices and edges.

            Now we define $\lambda$. The idea is that each $F_{i,j}$ will be active during its own dedicated time window. Formally, 
            define $\Delta_{i,j}=[f_{i,j}+1,f_{i,j}+m]$, for each pair $i,j \in [k]$, $i\neq j$, in a way that $\Delta_{i,j}\cap \Delta_{i',j'}=\emptyset$ whenever $\{i,j\}\neq \{i',j'\}$. 
            Now, consider $i,j \in [k]$ with $i\neq j$. For each $\ell \in \{0,\ldots,m-1\}$, we define $\lambda$ in a way that every value in $\Delta_{i,j}$ appears in some edge with endpoints $v^{i,j}_{\ell}v^{i,j}_{\ell+1}$. 
            Also, for each $\ell\in [m]$, we let $\lambda(v^{i,j}_{m-1+\ell} v^{i,j}_{m+\ell})=\{f_i+\ell\}$, and
            we define $\lambda(E(P^{i,j}_\ell))$ to be equal to $\Delta_{i,j}\setminus \{f_{i,j}+\ell \}$ and in a way that $P^{i,j}_\ell$ is a temporal $v^{i,j}_\ell,w^{i,j}_\ell$-path. 
            Finally, the only edges that remain unlabelled are the edges between $z$ and vertices of type $w$. For such edges, we reserve a time window of size $n = \lvert V(G)\rvert$, that we denote by  $\Delta_V$, where any timestep in such set is greater that any timestep we used to define $\lambda$ so far. Moreover, we associate each vertex $v \in V(G)$ with a timestep $t_v \in \Delta_V$. Let $w^{i,j}_{\ell} \in W$ and recall that such vertex is associated with $e^{i,j}_{\ell}\in E(G)$. Suppose $e^{i,j}_\ell$ have endpoints $xy$, and let the two edges of $G'$ with endpoints $w^{i,j}_{\ell}z$ be active in timesteps $\{t_x,t_y\}$.

            Now, we prove that $G$ has a clique of size $k$ if and only if $(G',\lambda)$ has a \snapcut{s,z} of size at most ${k \choose 2}+k$.
            Consider first a clique $C$ of $G$ of size $k$, and let $\{e^{i_1,j_1}_{\ell_1},\dots, e^{i_a,j_a}_{\ell_a}\}$ be the set of edges of $G$ between vertices of $C$. Notice that, because $C$ has a vertex from each part, we get that $a = {k \choose 2}$. Define $S = \{f_{i_b,j_b} + \ell_b \mid b\in \{1,\ldots,a\}\}\cup \{t_v \mid v \in C\}$. We prove that $S$ is a \snapcut{s,z}.
            By contradiction, suppose that $P$ is a temporal $s,z$-path not passing by $S$, i.e., such that $\lambda(E(P)) \cap S=\emptyset$. 
            Since $a = {k\choose 2}$ and all edges incident to $s$ are active in timesteps $\bigcup_{i,j\in[k],i\neq j}\Delta_{i,j}$, we can define
            $b\in [a]$ to be the index related to the first edge in $P$, i.e., $P$ starts in an edge of $F_{i_b,j_b}$, say the one active in timestep $f_{i_b,j_b} +\ell_b$. Observe that the value $f_{i_b,j_b} +\ell_b$ is within the temporal $s,z$-path contained in $F_{i_b,j_b}$, and that it also separates $s$ and $w^{i_b,j_b}_{\ell}$ for every $\ell \in [m]\setminus \{\ell_b\}$. Hence, $P$ must start with the temporal $s,w^{i_b,j_b}_{\ell_b}$-path contained in $F_{i_b,j_b}$. However, as $e^{i_b,j_b}_{\ell_b}$ is incident to vertices of the clique, say $x$ and $y$, then we have that $P$ uses timestep $t_x$ or $t_y$, a contradiction as $\{t_x,t_y\}\subseteq S$. 

            Now, suppose that $S$ is a minimum \snapcut{s,z} in $(G',\lambda)$ and that it has size at most ${k \choose 2}+k$. Let $V_S = \{x\in V(G)\mid t_x\in S\}$. We prove that $V_S$ is a clique of $G$ of size $k$. 
            Denote by $O$ the set of pairs $\{(i,j)\mid i,j\in [k], i<j\}$. We say that $(i,j)\in O$ is \emph{open} if $\Delta_{i,j}\cap S = \{f_{i,j}+\ell\}$ for some $\ell\in [m]$, and we say that $e^{i,j}_\ell$ is the \emph{open edge of $(i,j)$}. 
            The following simple facts will be useful:
            
            \begin{enumerate}
              \item\label{fact1} For every $i,j\in[k]$, $i\neq j$, we have $\Delta_{i,j}\cap S\neq \emptyset$: this is due to the fact that there is a \tpath{s,z} using only timesteps in $\Delta_{i,j}$; 
              
              \item\label{fact2} If $\ell\in[m]$ is such that $\Delta_{i,j}\cap S = \{f_{i,j}+\ell\}$, then $\{x,y\}\subseteq V_S$, where $xy$ are the endpoints of $e^{i,j}_\ell$: this is because there exists a \tpath{s,w^{i,j}_\ell} not using any timestep in $S$, and hence such path can be extended to a \tpath{s,z} by using an edge with endpoints $w^{i,j}_\ell$ either in timestep $t_x$ or in timestep $t_y$;
              
              \item\label{fact3} For every $i,j\in[k]$, $i\neq j$, we have $\lvert\Delta_{i,j}\cap S\rvert \le 2$: it suffices to see that any two timesteps in $\Delta_{i,j}$ intersects all temporal paths between $s$ and any vertex in $\{w^{i,j}_\ell\mid \ell\in [m]\}\cup \{z\}$;
              
              \item\label{fact4} If $x\in V_S$, then $x$ is incident to some open edge: indeed, if $x$ is not incident to any open edge, then $w^{i,j}_\ell$ is separated from $s$ by $S\setminus \{t_x\}$ for every edge $e^{i,j}_\ell$ incident in $x$, and since timestep $t_x$ contains only edges incident to some such $w^{i,j}_\ell$, it follows that $S\setminus \{t_x\}$ is also a \snapcut{s,z}, contradicting the minimality of $S$. 
          \end{enumerate}

          By Fact~\ref{fact3}, if $(i,j)$ is not open, then $\Delta_{i,j}\cap S = \{f_{i,j}+\ell_1, f_{i,j}+\ell_2\}$ for some pair of values $\ell_1,\ell_2\in [m]$. In such case, we say that edges $e^{i,j}_{\ell_1}$ and $e^{i,j}_{\ell_2}$ are \emph{chosen} for $(i,j)$. We show how to modify $S$ in order to decrease the number of chosen edges. Because of space constraints, the proof of the following claims are presented in Appendix~\ref{app:claims}.

          \begin{claim}\label{claim1}
          If $(i,j)$ is not open, then we can suppose that $V_S\cap (X_i\cup X_j) = \emptyset$.
          \end{claim}

          \begin{claim}\label{claim2}
           We can suppose that every pair is open.
          \end{claim}

          Finally, observe that the set of open edges, $E^*$, contains exactly ${k\choose 2}$ edges, by definition of open edge and by Claim~\ref{claim2}. We then get that $\lvert S\rvert = \lvert E^*\rvert + \lvert V_S\rvert = {k\choose 2} + \lvert V_S\rvert$. It follows that $\lvert V_S\rvert \le k$. Additionally, by Fact~\ref{fact2} we know that $E^*$ forms a subgraph of $G$ with vertex set $V_S$. Because $G$ is a simple graph, $E^*$ contains ${k\choose 2}$ edges, and $V_S$ contains at most $k$ vertices, the only way this can be possible is if $V_S$ contains exactly $k$ pairwise adjacent vertices, i.e., $V_S$ is a clique of size $k$, as we wanted to prove.
            
        \end{proof}


\section{Characterization and recognition of Mengerian graphs}\label{sec:menger}
 Let $\mathcal{P}$ be a set of \snapdisjoint\space\tpath{s,z}s and $S$ a \snapcut{s,z}. By definition, for each path  $P \in \mathcal{P}$, there is an edge in $P$ active at a timestep $\alpha_P$, for some $\alpha_P\in S$. As the paths in $\mathcal{P}$ are \snapdisjoint\space, we have that $\alpha_{P} \neq \alpha_Q$ for every $Q \in \mathcal{P}$ different from $P$. Therefore $|\mathcal{P}|\leq |S|$ and the inequality \maxsnap{G,\lambda}{s,z} $\leq$ \minsnap{G,\lambda}{s,z} follows. In Proposition~\ref{prop:tstrict} we show that this inequality can be strict, and later on we prove Theorem~\ref{thm:Menger}. 
 In our characterization, we have~$5$ graphs as forbidden structures, $M_1,M_2,M_3,M_4,M_5$, that are represented in Figure~\ref{fig:todos}. Let $\mathcal{M}$ be the set of such graphs.

In Figure~\ref{fig:tbadlabeling}, we present timefunctions for the graphs in $\mathcal{M}$ that turn the inequality \maxsnap{G,\lambda}{s,z} $\leq$ \minsnap{G,\lambda}{s,z} strict. This is formally stated in the next proposition and it is proved in Appendix~\ref{app:tstrict}. 

	\begin{figure}[h]
		\centering
			\includegraphics[scale=0.65]{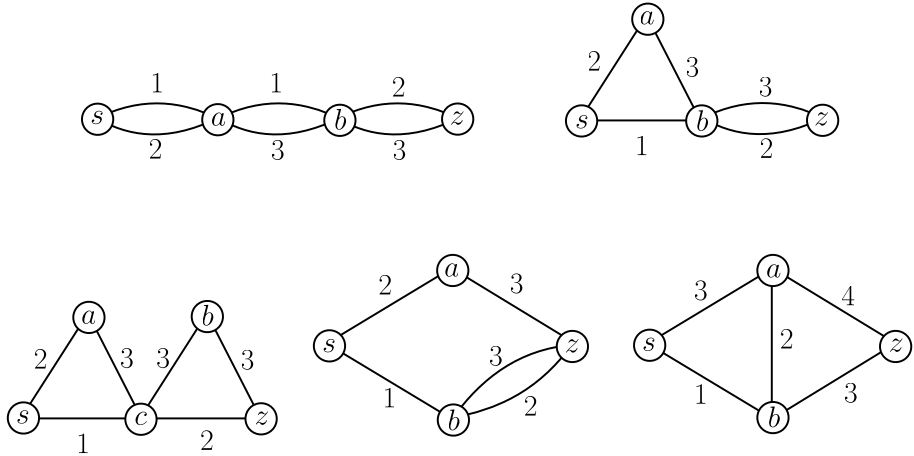}
		\caption{Graphs in the set $\mathcal{M}$ with timefunctions such that \maxsnap{G\lambda}{s,z} $<$ \minsnap{G,\lambda}{s,z}.}
		\label{fig:tbadlabeling}
	\end{figure}
	
\begin{proposition}\label{prop:tstrict}
    Let $(G,\lambda)$ be one of the temporal graphs depicted in Figure~\ref{fig:tbadlabeling}. Then \maxsnap{G,\lambda}{s,z} $<$ \minsnap{G,\lambda}{s,z}.
\end{proposition}



We now prove that the equality between the parameters always holds if $\lambda$ is injective.

\begin{proposition}\label{prop:lambdainjective}
    Let $(G,\lambda)$ be a temporal graph such that $\lambda$ is injective (i.e. $|\lambda^{-1}(\alpha)|\leq 1$ for every $\alpha \in \mathbb{N}$). Then \maxsnap{G,\lambda}{s,z} = \minsnap{G,\lambda}{s,z} for every $s,z \in V(G)$. Moreover, we can compute such value in polynomial time.
\end{proposition}

\begin{proof}
     If $P$ and $Q$ are \snapdisjoint\space\tpath{s,z}s, then, for $e \in E(P)$ and $f \in E(Q)$, we have that $\lambda(e)\neq \lambda(f)$, therefore $e\neq f$.  On other hand, if $P$ and $Q$ are such that $E(P)\cap E(Q)= \emptyset$, then we have that $\lambda(E(P))\cap \lambda(E(Q)) = \emptyset$. Thus, $P$ and $Q$ are \snapdisjoint\space if and only if they are edge disjoint. Therefore, the maximum size of edge disjoint \tpath{s,z}s is equal \maxsnap{G,\lambda}{s,z}. Moreover notice that $S \subseteq E(G)$ is a set such that every \tpath{s,z} uses an edge of $S$, then every \tpath{s,z} uses an edge active at timestep $\{\lambda(e) \mid e \in S\}$. If $S^T \subseteq [\tau]$ is such that every \tpath{s,z} uses an edge active at timestep $\alpha \in S^T$, then it uses the only edge in $\lambda^{-1}(\alpha)$. Therefore, the size of a minimum set of edges such that every \tpath{s,z} intersects such set is equal to \minsnap{G,\lambda}{s,z}. Using a result proved in~\cite{B.96} we conclude that \maxsnap{G,\lambda}{s,z}$=$\minsnap{G,\lambda}{s,z}, and that both parameters can be found in polynomial time.
\end{proof}

\begin{proof}[Proof of characterization]
We first prove that if $G$ has $M$ as a m-topological minor, for some $M\in {\cal M}$, then $G$ is not Mengerian. In Proposition~\ref{prop:tstrict}, we already get that each $H\in {\cal H}$ is not Mengerian. To finish necessity, we need to prove that adding vertices, adding edges and m-subdividing edges in $H$ does not lead to a Mengerian graph.  Propositions~\ref{prop:nonmeng} and~\ref{prop:subgraphClosed} do this work.

    Now, suppose by contradiction that $G$ does not contain any $M$ in $\mathcal{M}$ as m-topological minor and that $G$ is non-Mengerian. Hence, there exists $\lambda \colon E(G) \to \N$ and $s,z \in V(G)$ such that \maxsnap{G,\lambda}{s,z}$<$ \minsnap{G,\lambda}{s,z}. Suppose, without loss of generality, that among all such graphs and timefunctions we choose those that minimize $|V(G)|+|E(G)| + \tau(\lambda)$. 
    We first show that there are no edges connecting $s$ and $z$. Suppose otherwise, that there is an edge $e$ with endpoints and let $\lambda(e) = \alpha$. Let $G'$ be the graph obtained from removing all edges in $\lambda^{-1}(\alpha)$ and $\lambda'$, the restriction of $\lambda$ to $E(G)\setminus \lambda^{-1}(\alpha)$. Because $(s,e,z)$ is a \tpath{s,z}, observe that \maxsnap{G',\lambda'}{s,z}$=$\maxsnap{G,\lambda}{s,z}$-1$ and \minsnap{G',\lambda'}{s,z}$=$\minsnap{G,\lambda}{s,z}$-1$. This contradicts the fact that $(G,\lambda)$ minimizes $|V(G)|+|E(G)| + \tau(\lambda)$. 
    Now, suppose that there is a cycle $C$ containing $s$ and $z$. As $sz \notin E(G)$, we have that such cycle has size at least~4. As $G$ has no $M_5$ as m-topological minor, there are no paths between the vertices of $C$ that are not contained in $C$. Also, as $G$ has no $M_4$ as m-topological minor, all the multiedges of $C$ have multiplicity~1. In other words, the 2-connected component containing $C$ is formed just by $C$ itself, with each multiedge of C having multiplicity~1. In particular, there are only two paths connecting $s$ and $z$. If these two paths are \snapdisjoint\space\tpath{s,z}s, we take the timesteps of the two edges incident to~$z$ to obtain a \snapcut{s,z}\space as it kills the only two \tpath{s,z}s. This is a contradiction as we then get \maxsnap{G,\lambda}{s,z}$=$ \minsnap{G,\lambda}{s,z}. So, if both paths between $s$ and $z$ are temporal paths, they must intersect at a timestep $\alpha$. But in this case, the set $\{\alpha\}$ would be a \snapcut{s,z}, again a contradiction. 
    Therefore,  we can suppose that there is no cycle containing $s$ and $z$. This means that, if we consider the decomposition of $G$ in bi-connected components $B_1,\ldots,B_k$, then we have that $s$ and $z$ are in different components. Moreover, as we are supposing $(G,\lambda),s,z$ that minimize $|V(G)|+|E(G)| + \tau(\lambda)$, the graph induced by the decomposition is a path and we can suppose that $s \in V(B_1)$ and $z \in V(B_k)$. Suppose that there is $i \in [k]$, such that $B_i$ contains at least~3 vertices.  Let $v \in V(B_i)\cap V(B_j)$ for some $j \in \{i-1,i+1\}$ and $C_1$ be a cycle of $B_i$ containing~$v$. If $B_j$ has at least~3 vertices, then we can find another cycle $C_2$ contained in $B_j$ such that $V(C_1)\cap V(C_2)= \{v\}$, this contradicts the fact that $G$ has no $M_3$ as m-topological minor. So, we can suppose that $|V(B_j)|= 2$. If the multiedge contained in $B_j$ has multiplicity~1, then let $\alpha$ be the timestep of such edge. We have that $\{\alpha\}$ is a \snapcut{s,z}, and then \maxsnap{G,\lambda}{s,z} = \minsnap{G,\lambda}{s,z} a contradiction. 
    Therefore we can assume that the multiedge of $B_j$ has multiplicity at least~2; however the graph induced by such multiedge and $C_1$ is an m-subdivision of the graph $M_2$, again a contradiction. 
    Thus, we can assume that bi-connected component $B_i$ contains a cycle, for every $i \in \{1,\ldots,k\}$. In other words, $U(G)$ is an $s,z$-path. If some of the multiedges of $G$ has multiplicity~1, then \minsnap{G,\lambda}{s,z} $\leq 1$ and we have the equality, a contradiction. We can say also that $|U(G)|=3$, as otherwise $G$ would have $M_1$ as m-topological minor. Now, we show that in such case $G$ must be Mengerian, thus finishing the proof.
    
    We show that, for a graph $H$ such that $U(H)$ is a path of size $3$ between $s$ and $z$ and any timefunction $\lambda$ in $H$, we have that \maxsnap{G,\lambda}{s,z} = \minsnap{G,\lambda}{s,z}. Let $V(H)=\{s,w,z\}$, which means that the multiedges of $H$ are $sw$ and $wz$. We use induction on the number of edges of $H$. The base of induction is when $sw$ and $wz$ both have multiplicity~1. 
    Then, either $\lambda(sw)\leq \lambda(wz)$, in which case \maxsnap{H,\lambda}{s,z}$=$ \minsnap{H,\lambda}{s,z}$=1$, or $\lambda(sw) > \lambda(wz)$, in which case \maxsnap{H,\lambda}{s,z}$=$ \minsnap{H,\lambda}{s,z}$=0$. 
    Now, suppose valid when $\lvert E(H')\rvert\le m$ and consider $\lvert E(H)\rvert = m+1$. By Proposition~\ref{prop:lambdainjective}, we can suppose that there are two edges appearing at same timestep $\alpha$, say $f$ and $g$. As each snapshot is a simple graph, we get that $f,g$ form a \tpath{s,z}. 
    Let $H' = H-\{f,g\}$. By induction hypothesis, we have that \maxsnap{H',\lambda}{s,z} = \minsnap{H',\lambda}{s,z}. Now let $\mathcal{P}$ be a set of \snapdisjoint\space \tpath{s,z}s and $S$ be a \snapcut{s,z} in $(H',\lambda)$. Then, $\mathcal{P}\cup \{(s,f,w,g,z)\}$ is a set of \snapdisjoint \tpath{s,z}s in $(H,\lambda)$ and $S\cup\{\alpha\}$ is a \snapcut{s,z} in $(H,\lambda)$. Therefore, \minsnap{H,\lambda}{s,z} = \minsnap{H',\lambda}{s,z} + 1= \maxsnap{H',\lambda}{s,z}+ 1 = \maxsnap{H,\lambda}{s,z}.
\end{proof}

Now we turn our attention to the recognition of Mengerian graphs, showing that it can be done in polynomial time. We observe that the proof of characterization of Mengerian graphs helps us to construct an algorithm of recognition of Mengerian graphs. We make the proper adaptation and prove the next theorem in Appendix~\ref{app:rec}.

\begin{theorem}\label{theorem:rec}
One can decide in polynomial time whether a graph $G$ has a graph in $\mathcal{K}$ as m-topological minor.
\end{theorem}


\section{Multiedge Cut}\label{sec:multiedge_cut}

Finally, in this section we study another version of the cut problem. Recall that two 
\tpath{s,z}s are multiedge disjoint if they do not share any multiedge, and that a multiedge temporal $s,z$-cut is a set $S$ of multiedges of $G$ intersecting every \tpath{s,z}. In this section, we investigate the following problem.

\vspace{0.3cm}

        \noindent $\le h$-\textsc{Multiedge Temporal s,z-cut}\\
        \textbf{Input.} A temporal graph $(G,\lambda)$, vertices $s,z\in V(G)$, and an integer $h$.\\
        \textbf{Question.} Is there a multiedge temporal $s,z$-cut in $(G,\lambda)$ of size at most $h$?

\begin{theorem}\label{thm:multiedgecut}
    $\le h$-\textsc{Multiedge Temporal s,z-cut} is $\NP$-complete, even if $\tau=2$.
\end{theorem}
\begin{proof}
We make a reduction from {\textsc Vertex Cover}, which consists of, given a simple graph $G$ and a positive integer $k$, deciding whether there exists a subset $S\subseteq V(G)$ such that $|S|\le k$ and every $e\in E(G)$ is incident to some $u\in S$; such a set is called a \emph{vertex cover} (of size at most $k$). So, consider $\mathcal{I} = (G, k)$ an instance of {\textsc Vertex Cover}. We construct a graph $G'$ with vertex set $V(G') = \{s,z\}\cup \{x^1_{v},x^2_{v},x^3_{v},x^4_{v} \colon v \in V(G)\}\cup \{f_{vw} \mid vw \in E(G)\}$. One can use Figure~\ref{fig:multicut} to follow the construction. Then, we add edges from $s$ to $x^1_v$ and $x^2_v$, and from $x^3_v$ and $x^4_v$ to $z$, for every $v\in V(G)$. Also, let $(x^1_v,x^2_v,x^3_v,x^4_v)$ form a path, and add, for each edge $vw\in E(G)$, edges $x^1_v f_{vw}$ and $f_{vw} x^4_w$. More formally, we have:
\[\begin{array}{rl}
  E(G') =   & \{sx^1_{v},sx^2_{v},x^3_{v}z, x^4_{v}z \mid v \in V(G)\}  \\
     & \cup \{x^{i}_v x^{i+1}_{v} \mid i \in \{1,2,3\}, v \in V(G)\} \\
     & \cup \{x^1_{v} f_{vw}, f_{vw}x^{4}_{w} \mid vw \in E(G)\}
\end{array} .\] 

Finally, add a second edge with endpoints $x^2_vx^3_v$, for each $v\in V(G)$. Since these are the only edges with multiplicity greater than~1, we will generally denote an edge by its endpoints, with the exception of these, which we denote by $e^1_v,e^2_v$. Now for each $i \in \{1,\ldots,4\}$ define $X_{i} = \{x^i_v \colon v \in V\}$, let $F=\{f_vw \mid vw \in E(G)\}$ and consider a timefunction $\lambda$ such that:

\[ \lambda(e) =\left\{\begin{array}{ll}
1 & \mbox{, if } e\in (\{s\})\times X_1)\cup (X_1 \times X_2) \cup (X_3 \times \{z\}) \cup (X_{1}\times F)\\
2 & \mbox{, if $e\in (\{s\}\times X_2) \cup (X_3 \times X_4) \cup (X_4 \times \{z\})\cup (F\times X_4)$, and}\\
i & \mbox{, if $e = e^i_v$ for some $v\in V(G)$.}
\end{array}\right. \]

\begin{figure}[h]
		\centering
			\includegraphics[scale=0.6]{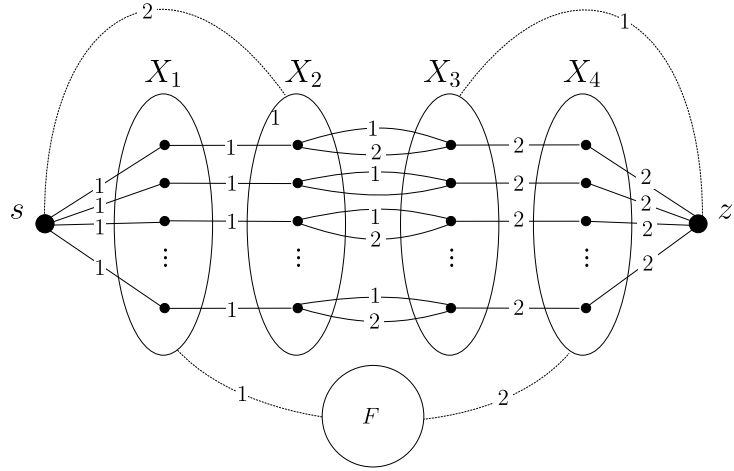}
		\caption{A representation of temporal graph $(G',\lambda)$. The dotted lines represent all the edges between the corresponding set of vertices. }
		\label{fig:multicut}
	\end{figure}

Because of space constraints, we present the proof of correctness in Appendix~\ref{app:multiedgecut}.

\end{proof}


        \section{Conclusion}\label{sec:conclusion}

        We have introduced the concept of snapshot disjointness and proved that the related paths and cut problems, when parameterized by the size of the solution, are both $\W[1]$-hard and $\XP$-time solvable. We then adapted to our context the definition of Mengerian graph introduced by Kempe, Kleinberg and Kumar~\cite{KKK.00}, giving also a characterization in the lines of the ones given in~\cite{KKK.00} and~\cite{IS.arxiv}, as well as a polynomial-time recognition algorithm. Since all our results concern only non-strict temporal paths, one can ask whether they also hold for strict paths.

Further open problems can be extracted from Table~\ref{table:results}. In particular, we ask whether the results for node departure disjoint paths and cuts presented in~\cite{MMS.19} for strict paths also hold for non-strict paths.

Finally, while Menger's Theorem is known to hold for edge disjoint paths~\cite{B.96,MMS.19}, it is also known not to hold for multiedge disjoint paths~\cite{IS.arxiv}. We reinforce the question posed in~\cite{IS.arxiv} about the characterization of Mengerian graphs in the context of multiedge disjoint paths.

        \bibliography{ref.bib}

\begin{thebibliography}{10}

\bibitem{AFG.17}
Amir Afrasiabi~Rad, Paola Flocchini, and Joanne Gaudet.
\newblock Computation and analysis of temporal betweenness in a knowledge
  mobilization network.
\newblock {\em Computational social networks}, 4(1):1--22, 2017.

\bibitem{B.etal.10}
Georg Baier, Thomas Erlebach, Alexander Hall, Ekkehard K{\"o}hler, Petr Kolman,
  Ond{\v{r}}ej Pangr{\'a}c, Heiko Schilling, and Martin Skutella.
\newblock Length-bounded cuts and flows.
\newblock {\em ACM Transactions on Algorithms (TALG)}, 7(1):1--27, 2010.

\bibitem{B.96}
Kenneth~A Berman.
\newblock Vulnerability of scheduled networks and a generalization of
  {M}enger's {T}heorem.
\newblock {\em Networks: An International Journal}, 28(3):125--134, 1996.

\bibitem{CFQS.12}
Arnaud Casteigts, Paola Flocchini, Walter Quattrociocchi, and Nicola Santoro.
\newblock Time-varying graphs and dynamic networks.
\newblock {\em International Journal of Parallel, Emergent and Distributed
  Systems}, 27(5):387--408, 2012.

\bibitem{CHMZ.19}
Arnaud Casteigts, Anne{-}Sophie Himmel, Hendrik Molter, and Philipp Zschoche.
\newblock Finding temporal paths under waiting time constraints.
\newblock In {\em 31st International Symposium on Algorithms and Computation,
  {ISAAC} 2020}, volume 181 of {\em LIPIcs}, pages 30:1--30:18, 2020.
\newblock \href {https://doi.org/10.4230/LIPIcs.ISAAC.2020.30}
  {\path{doi:10.4230/LIPIcs.ISAAC.2020.30}}.

\bibitem{param.book}
Marek Cygan, Fedor~V Fomin, {\L}ukasz Kowalik, Daniel Lokshtanov, D{\'a}niel
  Marx, Marcin Pilipczuk, Micha{\l} Pilipczuk, and Saket Saurabh.
\newblock {\em Parameterized algorithms}, volume~5.
\newblock Springer, 2015.

\bibitem{FMNR.22}
Eugen F\"{u}chsle, Hendrik Molter, Rolf Niedermeier, and Malte Renken.
\newblock {Temporal Connectivity: Coping with Foreseen and Unforeseen Delays}.
\newblock In James Aspnes and Othon Michail, editors, {\em 1st Symposium on
  Algorithmic Foundations of Dynamic Networks (SAND 2022)}, volume 221 of {\em
  Leibniz International Proceedings in Informatics (LIPIcs)}, pages
  17:1--17:17, Dagstuhl, Germany, 2022. Schloss Dagstuhl -- Leibniz-Zentrum
  f{\"u}r Informatik.
\newblock URL: \url{https://drops.dagstuhl.de/opus/volltexte/2022/15959}, \href
  {https://doi.org/10.4230/LIPIcs.SAND.2022.17}
  {\path{doi:10.4230/LIPIcs.SAND.2022.17}}.

\bibitem{GT.11}
Petr~A Golovach and Dimitrios~M Thilikos.
\newblock Paths of bounded length and their cuts: Parameterized complexity and
  algorithms.
\newblock {\em Discrete Optimization}, 8(1):72--86, 2011.

\bibitem{Holme.15}
Petter Holme.
\newblock Modern temporal network theory: a colloquium.
\newblock {\em The European Physical Journal B}, 88(9):234, 2015.

\bibitem{ILMS.arxiv}
Allen Ibiapina, Raul Lopes, Andrea Marino, and Ana Silva.
\newblock Menger's theorem for temporal paths (not walks).
\newblock {\em ArXiv}, arXiv:2206.15251, 2022.

\bibitem{IS.arxiv}
Allen Ibiapina and Ana Silva.
\newblock Mengerian graphs: Characterization and recognition.
\newblock {\em arxiv}, arXiv:2208.06517, 2022.

\bibitem{IPS.82}
Alon Itai, Yehoshua Perl, and Yossi Shiloach.
\newblock The complexity of finding maximum disjoint paths with length
  constraints.
\newblock {\em Networks}, 12(3):277--286, 1982.

\bibitem{KKK.00}
David Kempe, Jon Kleinberg, and Amit Kumar.
\newblock Connectivity and inference problems for temporal networks.
\newblock {\em Journal of Computer and System Sciences}, 64:820--842, 2002.

\bibitem{LVM.18}
Matthieu Latapy, Tiphaine Viard, and Cl{\'e}mence Magnien.
\newblock Stream graphs and link streams for the modeling of interactions over
  time.
\newblock {\em Social Network Analysis and Mining}, 8(1):61, 2018.

\bibitem{LMS.90}
Chung-Lun Li, S~Thomas McCormick, and David Simchi-Levi.
\newblock The complexity of finding two disjoint paths with min-max objective
  function.
\newblock {\em Discrete Applied Mathematics}, 26(1):105--115, 1990.

\bibitem{MMS.19}
George~B Mertzios, Othon Michail, and Paul~G Spirakis.
\newblock Temporal network optimization subject to connectivity constraints.
\newblock {\em Algorithmica}, 81(4):1416--1449, 2019.

\bibitem{SP.78}
Yossi Shiloach and Yehoshua Perl.
\newblock Finding two disjoint paths between two pairs of vertices in a graph.
\newblock {\em Journal of the ACM (JACM)}, 25(1):1--9, 1978.

\bibitem{west.book}
Douglas~Brent West et~al.
\newblock {\em Introduction to graph theory}, volume~2.
\newblock Prentice hall Upper Saddle River, 2001.

\bibitem{XFJ.03}
B~Bui Xuan, Afonso Ferreira, and Aubin Jarry.
\newblock Computing shortest, fastest, and foremost journeys in dynamic
  networks.
\newblock {\em International Journal of Foundations of Computer Science},
  14(02):267--285, 2003.

\bibitem{ZFMN.20}
Philipp Zschoche, Till Fluschnik, Hendrik Molter, and Rolf Niedermeier.
\newblock The complexity of finding small separators in temporal graphs.
\newblock {\em J. Comput. Syst. Sci.}, 107:72--92, 2020.
\newblock \href {https://doi.org/10.1016/j.jcss.2019.07.006}
  {\path{doi:10.1016/j.jcss.2019.07.006}}.

\end{thebibliography}

\appendix
\section{Proof of Propositions~\ref{prop:nonmeng} and~\ref{prop:subgraphClosed}}\label{app:propositions}


\begin{proof}[Proof of Proposition~\ref{prop:nonmeng}]
    Let $G$ be a non-Mengerian graph and consider $\lambda \in E(G) \to \N$ and $s,z \in V(G)$ to be such that \maxsnap{G,\lambda}{s,z}$<$ \minsnap{G,\lambda}{s,z}. Also, suppose that $H$ is obtained from $G$ by m-subdividing a multiedge, say $xy$. We construct a function $\lambda'$ from $\lambda$ that proves that $H$ is also non-Mengerian. 
		
    		Let $D\subseteq E(G)$ be the set of edges of $G$ with endpoints $xy$, and denote by $v_{xy}$ the vertex of $H$ created by the m-subdivision of $xy$.  Moreover, denote by $D_x$ and $D_y$ the sets of edges of $H$ with endpoints $xv_{xy}$ and $v_{xy}y$, respectively. Finally, define $\lambda'$ to be such that  $\lambda'(e) = \lambda(e)$, for every $e \in E(G)\setminus D$, and $\lambda'(D_x) = \lambda'(D_y) = \lambda(D)$. We show that \maxsnap{G,\lambda}{s,z}$=$\maxsnap{H,\lambda'}{s,z} and \minsnap{G,\lambda}{s,z}$=$\minsnap{H,\lambda'}{s,z}, which finishes our proof. 
    		
    		Given a set of \snapdisjoint \tpath{s,t}s in $(G,\lambda)$, if some of these paths, say $P$, uses the edge $xy$, then in $(H,\lambda')$ we can substitute such edge by an edge in $D_x$ and another in $D_y$ active at the same time to obtain a temporal path $P'$ such that $V(P')=V(P)\cup \{v_{xy}\}$. This gives us a set of \snapdisjoint\space\tpath{s,t}s in $(H,\lambda')$. In the other direction, if it is given a set of \snapdisjoint\space\tpath{s,t}s in $(H,\lambda')$, if some of them uses the vertex $v_{xy}$, let $f_j$ be the edge used in $D_j$ for $j \in \{x,y\}$. Suppose without loss of generality that $\lambda(f_x)\geq \lambda(f_y)$. Then we substitute both edges incident to $v_{xy}$ by an edge in $D$ appearing at time $\lambda(f_x)$. This implies that \maxsnap{G,\lambda}{s,z}$=$\maxsnap{G,\lambda'}{s,z}. To see that \minsnap{G,\lambda}{s,z}$=$\minsnap{H,\lambda'}{s,z} one can need to recall that $\lambda'(D_x)=\lambda'(D_y) = \lambda(D)$.
\end{proof}

 \begin{proof}[Proof of Proposition~\ref{prop:subgraphClosed}]
		To prove necessity, suppose that $H\subseteq G$ is non-Mengerian, and let $s,z,\lambda$ be such that \maxsnap{H,\lambda}{s,z}$<$ \minsnap{H,\lambda}{s,z}. Consider the timefunction  $\lambda'$  in $E(G)$ defined as follows. 
  
		\[\lambda'(e) = \left\{\begin{array}{ll}
			\lambda(e)+1 & \mbox{, for every } e\in E(H),\\
			1 & \mbox{, for every $e\in E(G)\setminus E(H)$ with endpoints $yt$, and}\\
			\max\lambda(E(H))+2 & \mbox{, otherwise.}
		\end{array}\right.\]

		Because $H\subseteq G$ and $\lambda\subseteq \lambda'$, note that we get \maxsnap{H,\lambda}{s,z}$\leq$\maxsnap{G,\lambda'}{s,z} and \minsnap{H,\lambda}{s,z}$\leq$ \minsnap{G,\lambda'}{s,z}. Therefore it suffices to prove \maxsnap{G,\lambda'}{s,z}$\leq$ \maxsnap{H,\lambda}{s,z} and \minsnap{G,\lambda'}{s,z}$\leq$\minsnap{H,\lambda}{s,z}. Indeed, these hold because the timefunction $\lambda'$ does not allow for the existence of a \tpath{s,t} not contained in $H$.
	\end{proof}

\section{Proof of Theorem~\ref{theorem:cuts}}\label{app:cuts}

            Let $(G,\lambda),s,z,h$ as in the hypothesis of the theorem. For each subset $S\subseteq [\tau]$ of size $h$, define $G_S$ such that $V(G_S)=V(G)$ and $E(G_S)=\{e \in E(G) \mid \lambda(e)\notin S\}$. Define also $\lambda_S(e)=\lambda(e)$ for all $e \in E(G_S)$.    
            Now, by the definition of $G_S$, any \tpath{s,z} in $(G_S,\lambda_S)$ is a \tpath{s,z} in $(G,\lambda)$ that does not use edges active at timesteps in $S$. Reciprocally, every \tpath{s,z} in $G$ that does not use edges active at timesteps in $S$ is a \tpath{s,z} in $(G_S,\lambda_S)$.
            As testing if there is a \tpath{s,z} in $(G_S,\lambda_S)$ can be done in polynomial time~\cite{XFJ.03} and $(G_S,\lambda_S)$ can be constructed in $O(n+m)$ time, it suffices to apply this test to $(G_S,\lambda_S)$ for every $S\subseteq [\tau]$ of size $h$.  Since there are at most $\tau^h$ such sets, the theorem follows. 

\section{Proof of correctness of reduction in Theorem~\ref{theo:snappaths}}\label{app:snappaths}

      We show that $(G',\lambda)$ has $k$ \snapdisjoint\space \tpath{s,z}s if and only if $G$ has an independent set of size at least $k$. By the definition of $G'$, all \tpath{s,z}s are of type $Q_u$ for some $u \in V(G)$. Therefore, it suffices to show that a subset $S\subseteq V(G)$ is an independent set of $G$ if and only if $\{Q_u \mid u \in S\}$ is a set of \snapdisjoint\space \tpath{s,z}s in $(G',\lambda)$. Suppose first that $S$ is an independent set, and suppose by contradiction that $u_1,u_2 \in S$ are such that $Q_{u_1}$ and $Q_{u_2}$ are not \snapdisjoint. Then there exists $e_1 \in E(Q_1)$ and $e_2\in E(Q_2)$ such that $\lambda(e_1)=\lambda(e_2)$. By construction, this means that $e_1 = e_2$, and since $u_1\neq u_2$, we get that actually this edge has endpoints $u_1u_2$, a contradiction as $S$ is an independent set. Thus, $\{Q_u \mid u \in S\}$ is a set of \snapdisjoint\space \tpath{s,z}s. Finally, observe that if $e_i = vw\in E(G)$, then $\lambda(Q_v)\cap \lambda(Q_w) = \{i\}$, which directly implies that if $\{Q_u \mid u \in S\}$ is a set of \snapdisjoint\space \tpath{s,z}s, then $S$ cannot contain any pair of adjacent vertices. 


\section{Proof of claims in Theorem~\ref{thm:snapcut_W1hard}}\label{app:claims}

          \begin{proof}[Proof of Claim~\ref{claim1}]
          Let $e^{i,j}_{\ell_1}$ and $e^{i,j}_{\ell_2}$ be the chosen edges for $(i,j)$. Denote by $x_1y_1$ and $x_2y_2$  the endpoints of $e^{i,j}_{\ell_1}$ and $e^{i,j}_{\ell_2}$, respectively. Suppose without loss of generality that $\{x_1,x_2\}\subseteq X_i$ and $\{y_1,y_2\}\subseteq X_j$. Now suppose that there exists $x\in X_i\cap S$. 
          If $x\in \{x_1,x_2\}$, say $x=x_1$, then replace $f_{i,j}+\ell_2$ in $S$ by $y_1$, obtaining $S'$. Observe that in this case $e^{i,j}_{\ell_1}$ becomes an open edge, but such that $\{t_{x_1},t_{y_1}\}\subseteq S'$. Hence $S'$ is still a \snapcut{s,z} containing fewer chosen edges. And if $x\notin \{x_1,x_2\}$, then let $xy$ be any edge incident in $x$ such that $y\in X_j$.  We can suppose that such edge exists as otherwise $x$ cannot be in any multicolored $k$-clique and then we can remove it from $G$. Let $\ell\in [m]$ be such that $e^{i,j}_\ell$ has endpoints $xy$. Observe that Fact~\ref{fact3} also tells us that $S' = (S\setminus \{f_{i,j}+\ell_2\})\cup \{f_{i,j}+\ell\}$ is a \snapcut{s,z}. We can then apply the previous argument to replace $f_{i,j}+\ell_1$ by $t_y$ to again obtain a \snapcut{s,z} with fewer chosen edges.
          \end{proof}

          \begin{proof}[Proof of Claim~\ref{claim2}]
          We can clearly suppose that $G$ is connected, as otherwise the answer to \textsc{Multicolored $k$-Clique} is trivially ``no''. 
          Now, let $I_1$ be the set of indices $\{i\in [k]\mid V_S\cap X_i = \emptyset\}$ and $I_2 = S\setminus I_1$. Suppose first that $I_1\neq \emptyset$ and $I_2\neq \emptyset$, and consider $i_1\in I_1$ and $i_2\in I_2$. By Claim~\ref{claim1} and because $i_2\in I_2$, we know that $(i_1,i_2)$ is open. So let $e^{i_1,i_2}_\ell$ be the open edge for $(i_1,i_2)$, and let $xy$ be its endpoints, with $x\in X_{i_1}$ and $y\in X_{i_2}$. By Fact~\ref{fact2}, we get that $\{t_x,t_y\}\subseteq S$, i.e., $\{x,y\}\subseteq V_S$, contradicting Claim~\ref{claim1} as $x\in X_{i_1}$ and $i_1\in I_1$. 

          So either $I_1=\emptyset$ or $I_2 = \emptyset$. Observe that if $I_1 = \emptyset$, then $V_S\cap X_i\neq\emptyset$, for every $i\in [k]$, and the claim follows from Claim~\ref{claim1}. And if $I_2 = \emptyset$, then by Fact~\ref{fact3} we get that $S$ contains two edges for every pair $i,j$, totalling $\lvert S\rvert = 2{k\choose 2}$. Since $\lvert S\rvert \le {k\choose 2}+2$, we get that this happens only if $k\le 3$, in which case \textsc{Multicolored $k$-Clique} is polynomial-time solvable.
          \end{proof}


\section{Proof of Proposition~\ref{prop:tstrict}}\label{app:tstrict}
    To observe that \minsnap{G,\lambda}{s,z}$ > 1$, one just needs do verify that for each timestep, there is a \tpath{s,z} not using this timestep. Now, for the cases $G \in \{M_1,M_2,M_3,M_4\}$, suppose \maxsnap{G,\lambda}{s,z}$=$\minsnap{G,\lambda}{s,z} and let $Q,J$ be \snapdisjoint\space\tpath{s,z}s. In each of those cases, there are only two edges incident to $s$, one active at timestep~1 and other at timestep~2. Hence, one of these paths, say $Q$, starts at timestep~2. Observe that in this case, $Q$ can only finish through edges active at timestep~3. Therefore, $J$ cannot use timestep~2 nor~3, however, all edges incident to~$z$ are active at timesteps~2 or~3, a contradiction as there is no \tpath{s,z} contained in the first snapshot. Finally, suppose $G = M_5$; we will apply a similar argument. So, let $Q,J$ be \snapdisjoint\space\tpath{s,z}s. Note that one of them, say $Q$, must use the edge incident to $s$ active at timestep~3 and, therefore, finishes using the edge incident to $z$ active at timestep~4. It follows that $J$ is not allowed to use timestep~$3$ nor~$4$, a contradiction as all edges incident to~$z$ are only active at such timesteps.


\section{Proof of Theorem~\ref{theorem:rec}}\label{app:rec}
    First, we find a decomposition of $G$ in 2-connected components, $B_1,\dots,B_k$. This can be done in $O(m+n)$ (see e.g.~\cite{west.book}). 

    We consider the set of all components $B_i$ such that $|V(B_i)|=2$ and the multiedge contained in $B_i$ has multiplicity at least~2, let $D$ be such set. Then, we test if some component $B_i$ in $D$ has vertex in common if a component $B_i$ of size at least~3. This takes $O(k^2)$ times. If the answer is positive, then we would have that $G$ has $M_2$ as m-topological minor. Then we can suppose the following:
    
\vspace{0.3cm}
 \noindent   \textbf{2. }No component in $D$ share a vertex with other component that has at least~3 vertices.
\vspace{0.3cm}

    Now, we separate all components that have at least three vertices and test if two of them share a vertex, this step takes $O(k^2)$ times. If the answer is true for some two components, then $G$ has $M_3$ as m-topological minor. Therefore, we can suppose:
    
\vspace{0.3cm}
  \noindent  \textbf{3. }Components $B_i$ and $B_j$ of size at least~3 do not share vertex.
\vspace{0.3cm}

    Now, we look to the components $B_i$ such that $|V(B_i)|\geq 4$. As $B_i$ is two connected, it has a cycle $C$ that we can find in $O(n^2)$. Then we check if there is an edge $e \in E(B_i)\setminus E(C)$. If the answer is positive, then as $B_i$ is 2-connected, there is a chord in $C$ containing $e$. This implies that $G$ has $M_5$ as m-topological minor. So, we have the following property.
    
\vspace{0.3cm}
   \noindent \textbf{4. }Each component $B_i$ of size at least~4 is such that $U(B_i)$ is a cycle.
\vspace{0.3cm}

    With such observation, we now can test if the components $B_i$ of size at least~4 contains some multiedge with multiplicity at least~2. If the answer is positive we are done as it would lead to $B_i$ having $M_4$ as m-topological minor. So, we can suppose otherwise:
    
\vspace{0.3cm}
   \noindent \textbf{5. }Each component $B_i$ of size at least~4 contains no multiedges of multiplicity at least~2.
\vspace{0.3cm}

    Now we show that the properties 2-5 implies that $G$ has no graphs in $\mathcal{M}$ as m-topological minors. One just need to observe that every cycle is contained in a 2-connected component. So the properties 2-5, assure us that $G$ has no $M_2$,$M_3$,$M_4$ or $M_5$ as m-topological minors.

    Now, we only need to test if $G$ has $M_1$ as m-topological minor. We can consider the graph $G'$ obtained from $G$ but excluding the multiedges with multiplicity~1. Then, we test if $G'$ has a path with at least~4 vertices. The answer is positive if and only if $G$ has $M_1$ as m-topological minor.  This finishes the recognition.


\section{Proof of correctness of reduction in Theorem~\ref{thm:multiedgecut}}\label{app:multiedgecut}

We prove that $G$ has a vertex cover of size at most $k$ if and only if $(G',\lambda)$ has a multiedge temporal $s,z$-cut of size at most $n+k$. Given a solution $S$ of {\textsc Vertex Cover}, we can define $S' = \{sx^{1}_{v}, x^{4}_{v}z \colon v \in S\} \cup \{x^{2}_{v}x^{3}_{v} \colon v \notin S\}$. It remains to argue that $S'$ separates $s$ from $z$, or more formally, that there is no \tpath{s,z} in $(G'-S^*,\lambda)$, where $S^*$ contains every $e\in E(G')$ such that the endpoints of $e$ are in $S'$. Notice that if a \tpath{s,z} does not use edges between $X_{1}$ and $F$, then it contains one of the following  paths:
$(s,x^1_v,x^2_v,x^3_v,x^4_v,z)$, $(s,x^2_v,x^3_v,x^4_v,z)$, or 
$(s,x^1_v,x^2_v,x^3_v,z).$ In any case such path intersects $S'$. Now suppose that a \tpath{s,z}, $P$, uses an edge $x^1_vf_{vw}$, which implies that it also uses $f_{vw}x^4_{w}$ and therefore it arrives in $z$ at timestep~2. Note that the only edges active at timeste~2 incident to $x^4_w$ are incident either to $F$, or to $x^3_{w}$, or to $z$. In the former case, we hit a dead end because $f_{uw}$ has only one edge incident to it in timestep~2, for every $u\in N(w)$. 
This also happens in the second case, since from $x^3_w$ one can only go to $x^2_w$, hitting again a dead end. We then get that if $P$ contains $x^1_vf_{vw}$, then it must also contain $x^4_w z$. 
By a similar argument one can also show that it must contain $s x^1_v$ too. As $vw \in E(G)$, at least one of $v$ and $w$ are in $S$, and by construction we get that $P$ uses some multiedge of $S'$, as we wanted to show.

Now let $S'$ be a multiedge temporal $s,z$-cut of size at most $n+k$. For each $v \in V(G)$, let the set of multiedges $\{sx^1_v,x^1_vx^2_v,x^2_vx^3_v,x^3_vx^4_v,x^4_vz\}$ be denoted by $A_v$, and let $S'_{v} = S' \cap A_{v}$. Because $A_v$ forms a \tpath{s,z}, we know that $S'_v\neq \emptyset$. One can notice that every \tpath{s,z} using $x^1_{v}x^2_{v}$ also uses $sx^{1}_{v}$, and in the same way every \tpath{s,z} using $x^3_vx^4_v$ uses $x^4_vt$. Then by changing $S'$ if necessary, we can suppose that $S'_{v}$ is a non-empty subset of $\{sx^{1}_v, x^2_{v}x^3_v, x^4_v z\}$. Now we show that we can actually suppose that $S'\cap A_v$ is either $\{sx^{1}_v, x^4_vt\}$ or $\{x^2_{v}x^3_v\}$. We do it by analysing the cases where this does not happen.

\begin{itemize}
    \item $S'_{v} = \{sx^{1}_v\}$ or $S'_{v} = \{x^{4}_v z\}$. We just solve the first subcase as the second is similar. Notice that $(s,x^2_v,x^3_v,x^4_v,z)$ is a \tpath{s,z}, and that the only edge in this path that can be in $S'$ is $sx^2_v$ because of the case being analyzed. So, removing $s x^{2}_v$ and adding $x^4_v$ to $S'$ maintains the property of being a multiedge temporal $s,z$-cut and turns $S'_{v}$ into $\{sx^1_v,x^4_v z\}$.
    
    \item $S'_{v} = \{sx^1_v, x^2_{v}x^3_v\}$ or $S'_v = \{x^2_{v}x^3_v, x^4_v z \}$. In both subcases we can remove $x^2_vx^3_v$ and add either $x^4 z$ (in the first case) or $sx^1_v$ (in the second one). 
\end{itemize}

Now we can define $S = \{v \in V \colon S_{v} = \{sx^{1}_v, x^4_vt\}\}$. The desired property $|S| \leq k$ follows from the fact that $1\le |S'_v|\le 2$ for every $v\in V(G)$, and that $|S'|\le n+k$. Finally, suppose that $vw \in E(G)$ is such that $S\cap \{v,w\}=\emptyset$. Then $(s,x^1_v,x^4_w,z)$ is a \tpath{s,z} not passing through the edges in $S'$, a contradiction.

\end{document}